\documentclass[11pt]{article}
\usepackage[T1]{fontenc}
\usepackage{lmodern}
							
\usepackage[frenchmath]{mathastext}
\usepackage{natbib}       									
\usepackage{amsmath}                        
\usepackage{graphicx} 										
\usepackage{caption}
\usepackage{booktabs}
\usepackage{subfigure}
\usepackage{float} 	
\usepackage{enumitem}										
\usepackage{mdwlist}												
\usepackage[dvipsnames]{xcolor}							
\usepackage[plainpages=false, pdfpagelabels]{hyperref} 
	\hypersetup{
		colorlinks   = true,
		citecolor    = RoyalBlue,
		linkcolor    = RubineRed, 
		urlcolor     = Turquoise
	}
\usepackage{url}
\usepackage{amssymb}                        
\usepackage[mathscr]{eucal}                 
\usepackage{dsfont}									
\usepackage[top=1.0in, bottom=1.0in, left=0.8in, right=0.8in]{geometry} 
\usepackage{ulem}
\usepackage{mathtools}
\usepackage{bbold}               
\mathtoolsset{showonlyrefs=true}       
\usepackage{amsthm}                
	\allowdisplaybreaks                       
	\theoremstyle{plain}
	\newtheorem{theorem}{Theorem}
	\newtheorem{lemma}[theorem]{Lemma}       
	\newtheorem{proposition}[theorem]{Proposition}
	
	\theoremstyle{definition}
	\newtheorem{definition}{Definition}

\newcommand\Eb{\mathds{E}}

\newcommand\Pb{\mathds{P}}

\newcommand\Rb{\mathds{R}}

\newcommand\Ac{\mathscr{A}}

\newcommand\Fc{\mathscr{F}}

\newcommand\Lc{\mathscr{L}}

\newcommand\Tc{\mathscr{T}}

\newcommand\Ebt{\widetilde{\Eb}}
\newcommand\Pbt{\widetilde{\Pb}}

\newcommand\Wt{\widetilde{W}}
\newcommand\Bt{\widetilde{B}}
\newcommand\Ct{\widetilde{C}}

\newcommand{\dd}{\mathrm{d}}

\newcommand{\eqlnostar}[2]{\begin{align}\label{#1}#2\end{align}}
\newcommand{\eqstar}[1]{\begin{align*}#1\end{align*}}

\usepackage{stackengine}
\usepackage{scalerel}
\usepackage{xcolor,amssymb}

\makeatother

\begin{document}

\title{Hedging longevity risk in defined contribution pension schemes}

\author{
Ankush Agarwal
\thanks{\textbf{e-mail}: \url{ankush.agarwal@glasgow.ac.uk}}
\qquad
Christian-Oliver Ewald
\thanks{\textbf{e-mail}: \url{christian.ewald@glasgow.ac.uk}}
\qquad 
Yongjie Wang
\thanks{\textbf{e-mail} (corresponding author): \url{y.wang.12@research.gla.ac.uk}}
\\
Adam Smith Business School, University of Glasgow, G12 8QQ Glasgow, United Kingdom
}

\date{May 2020}
\maketitle

\begin{abstract}
Pension schemes all over the world are under increasing pressure to efficiently hedge the longevity risk posed by ageing populations. In this work, we study an optimal investment problem for a defined contribution pension scheme which decides to hedge the longevity risk using a mortality-linked security, typically a longevity bond. The pension scheme invests in the risky assets available in the market, including the longevity bond, by using the contributions from a representative scheme member to ensure a minimum guarantee such that the member is able to purchase a lifetime annuity upon retirement. We transform this constrained optimal investment problem into a single investment portfolio optimisation problem by replicating a self-financing portfolio of future contributions from the member and the minimum guarantee provided by the scheme. We solve the resulting optimisation problem using the dynamic programming principle and through a series of numerical studies reveal that the longevity risk has an important impact on the performance of investment strategy. Our results add to the growing mathematical evidence supporting the use of mortality-linked securities for efficient hedging of the longevity risk.

{\bf Keywords:}
defined contribution pension scheme, longevity bond, stochastic control, dynamic programming principle
\end{abstract}

\section{Introduction}
\label{sec:introduction}

A pension scheme is an important economic mechanism in the society as it provides people with regular incomes after their retirement from the productive labour workforce. According to the benefit and contribution policy, there are two main categories of pension schemes: defined benefit schemes (DB schemes) and defined contribution schemes (DC schemes). In a DB scheme, pension benefits to be paid by the scheme after retirement time are pre-defined. In this case, scheme members only need to pay contributions regularly and bear no investment risk, while the scheme manager bears the risk of bad investment performance and may fail to deliver the benefits. In a DC scheme, the amount of contributions payable by scheme members is pre-determined instead of the benefit payments. The benefits depend on the size of the accumulated contributions and scheme manager's investment performance, and are uncertain until the retirement time. As opposed to the DB scheme members, members of a DC scheme have more choice on the ways to receive their benefits. Typically, there are three options: to purchase a lifetime annuity, or, to opt for a flexible income-drawdown option (that is, to withdraw money periodically while the remaining money stays in the pension scheme), or, to withdraw the lump sum cash amount from the scheme. In a DC scheme, the employer of scheme members bears no risk as its only responsibility is to pay the contributions along with its employees. However, the employees face the risk of receiving insufficient benefits after their retirement due to bad investment performance. \cite {biffis2014keeping} mentioned that the estimate of global amount of annuity- and pension-related longevity risk exposure amount to \$15 trillion. According to \cite{cocco2012longevity}, the average life expectancy of 65-year-old US (UK) males increases by 1.2 (1.5) years per decade. As a consequence, a DB scheme for those populations would have needed 29\% more wealth in 2007 than in 1970. Due to the unsustainability of the DB schemes, the DC schemes are becoming increasingly popular among the employers.

\cite{boulier2001optimal} studied the optimal investment problem for a DC scheme under stochastic interest rate framework in which a downside protection for the member's benefits is provided. They obtained the optimal investment strategy which maximises the expected terminal utility from the surplus between the scheme's final wealth and the downside guarantee by applying the dynamic programming principle. \cite{gao2008stochastic} used the dual approach to solve the optimal asset allocation problem for a DC scheme in a market with stochastic interest rate. \cite{deelstra2003optimal} studied the assets and liabilities management (ALM) problem for the DC pension schemes from the perspective of a scheme manager. They maximised the expected utility of terminal wealth in the presence of a specific minimum guarantee and used the martingale theory to solve the optimisation problem. \cite{han2012optimal} studied the optimal asset allocation problem for DC pension schemes considering the inflation risk and labour income risk. They introduced a minimum guarantee on the purchase of an inflation-indexed annuity at retirement. To hedge the inflation risk, they included an inflation-indexed bond in the investment portfolio.

Inspired by these previous works, we consider a DC scheme in which the scheme manager allocates the wealth in a financial market to achieve the amount needed to buy a lifetime annuity upon retirement of the members. The annuity acts as a minimum guarantee and its price depends on the expected remaining lifetime of the scheme members. Thus, it is crucial to model the members' mortality behaviour properly. The force of mortality, which is the instantaneous rate of mortality, is often used to implement survival analysis in actuarial science. Classical works of \cite{de1725annuities} and \cite{gompertz1825xxiv} have studied deterministic force of mortality models. However, more recent research on mortality modelling considers discrete-time and continuous-time models with stochastic force of mortality. It is straightforward to model the force of mortality in a discrete-time setting since the mortality data are usually reported annually. \cite{lee1992modeling} were among the earliest to model and estimate the force of mortality using time series methods. Other discrete-time models include, for example, the CBD model and Renshaw-Haberman cohort model (\cite{cairns2006two}; \cite{renshaw2006cohort}). Some studies, such as \cite{milevsky2001mortality} and \cite{dahl2004stochastic}, found similarities between interest rates and force of mortality; for example, that they are positive and have a term structure. Thus, drawing from the interest rate modelling literature, diffusion processes and jump processes are now used to study the force of mortality. In particular, affine mortality models are popular and are studied in works such as \cite{dahl2004stochastic}, \cite{biffis2006bidimensional}, \cite{luciano2005non} and \cite{russo2011calibrating}. In this work, we assume that the evolution of mortality rate of all the scheme members can be described by the same continuous-time stochastic process and thus, choose a representative member to study the problem. We follow \cite{menoncin2009death} and describe the force of mortality using an affine model which is an analogue of the Cox-Ingersoll-Ross (CIR) process. The analytical tractability of the affine model allows us to price mortality-linked securities using the arbitrage-free pricing framework that has been developed for interest-rate derivatives.

Proposed by \cite{blake2001survivor}, a longevity bond provides coupon payment based on the number of survivors in a chosen reference population. Therefore, investment in a longevity bond not only provides an efficient way to hedge the longevity risk, but also allows diversification of investment portfolios. \cite{menoncin2008role} studied an optimal consumption and investment problem for an investor with a stochastic death time. He maximised the investor's intertemporal consumption until the death time and used a rolling longevity bond to hedge against the investor's longevity risk. He showed that the optimal amount to be invested in the longevity bond decreases over time since the need for hedging longevity risk decreases as approaching death while the optimal consumption increases over time.  \cite{de2017existence} modelled the force of mortality using the CIR process which guarantees the mortality rates to be non-negative. They argued that although there is no liquid market for such longevity bonds, it is not practical to put the market price of longevity risk at zero. Instead, they assumed a time-varying market price of risk which is proportional to the square root of the mortality rate. \cite{cocco2012longevity} studied the optimal consumption and investment problem in a life-cycle model. By calibrating to the US historical data and current projections, they showed considerable uncertainty with respect to the future improvements in mortality rates. They also suggested that longevity linked securities can help in longevity risk management. \cite{menoncin2017longevity} studied the optimal consumption and investment problem for an individual investor to hedge his longevity risk before retirement. They showed that the optimal proportion that should be invested in longevity bond is higher than other assets. In this work, we consider a financial market that consists of three risky assets: a stock, a rolling bond and a rolling longevity bond. The results show that the longevity bond provides an efficient way to hedge the longevity risk.

Most articles studying the optimal portfolio strategy for DC schemes, focus on the financial risks (for example, interest rate risk, inflation risk) and ignore the longevity risk. Even the studies which take the longevity risk in account, mainly focus on optimal asset allocation problems for DB pension schemes and use time-varying but deterministic force of mortality models to measure the longevity risk. In addition, as the pension schemes are managed on behalf of its members, it makes more sense to study optimal investment  problems which maximise the members' expected terminal utility. However, only a few authors, for example, \cite{han2012optimal}, \cite{he2015optimal}, studied the optimal investment problems from the perspective of scheme members, while the rest focus mainly on the ALM framework. In our work, we study the optimal asset allocation problem for a DC scheme from the perspective of a representative scheme member under the framework of stochastic force of mortality and stochastic interest rate. In particular, our main contribution is in extending the works of \cite{boulier2001optimal}, \cite{gao2008stochastic}, \cite{menoncin2017longevity}, to investigate the optimal portfolio allocation for DC schemes while hedging the longevity risk. The representative member requires that the scheme's wealth level must be at least sufficient to buy a lifetime annuity which acts as a minimum guarantee. The scheme manager maximises the expected utility of the terminal surplus between the scheme's wealth level and the minimum guarantee. To hedge the longevity risk, a rolling longevity bond as introduced in \cite{menoncin2008role} is added to the investment portfolio. Our results show that the longevity risk plays an important role in the pension scheme's risk management and reveal that the longevity bond can not only offer an efficient way to hedge the future longevity risk, but can also provide attractive risk premiums.

The rest of this paper is organised as follows. In Section \ref{sec:framework} we present the mathematical framework of the problem and introduce different risky assets considered in the financial market. In Section \ref{sec:criterion}, we first formulate the constrained optimisation problem in which the scheme's wealth level should be sufficient to purchase a lifetime annuity at member's retirement. We identify different components of the investment portfolio and reformulate the portfolio selection as a single investment portfolio optimisation problem in Section \ref{sec:uncostrained}. We derive the analytical solution for the optimal investment strategy by using the dynamic programming principle in Section \ref{sec:solution}. Section \ref{sec:numerics} discusses several numerical studies including sensitivity analyses with respect to different model parameters which reveal the significance of introducing a rolling longevity bond in the investment portfolio.
\section{Financial market setting}
\label{sec:framework}
Let $(\Omega, \Fc, \{\Fc(t)\}_{t\ge 0}, \Pb)$ be a filtered probability space satisfying the usual conditions on an infinite time horizon $\Tc=[0,\infty)$. $\Pb$ is the physical (observable) probability measure and $\Fc (t)$ signifies the information available to the investor at time $t.$ On this probability space, we consider a frictionless financial market consisting of a stock, a \textit{rolling bond} and a \textit{rolling longevity bond}. For practical pricing of zero-coupon bond and longevity bond, we consider a stochastic risk-free interest rate $r(t)$ and a stochastic force of mortality $\lambda(t).$ Furthermore, we denote a three-dimensional standard Brownian motion under $\Pb$ by $\bigl\{W(t)\mid t\in \Tc\bigr\}=\bigl\{\bigl[W_1 (t),W_2 (t),W_3 (t)\bigr]^\prime \mid t\in \Tc\bigr\}.$ We assume that $r(t)$ is described by a CIR process:
\eqlnostar{eq:dr}{
\dd r(t)=(a_r - b_r r(t))\dd t+\sigma_r \sqrt{r(t)}\dd W_1 (t),&& r(0)=r_0,
}
where $a_r$, $b_r$ and $\sigma_r$ are positive constants. We further assume that the Feller's condition $2a_r > \sigma_r ^2$ is satisfied so that, for any $t\in \Tc$, $r(t)>0$ almost surely under $\Pb$.

As stated earlier in Section \ref{sec:introduction}, affine models are popular when modelling the stochastic force of mortality. \cite{luciano2005non} described the force of mortality (also called mortality intensity) by affine models and calibrated the models using observed and projected UK mortality tables. They claimed that affine processes with a deterministic part which increases exponentially could describe the evolution of force of mortality properly. \cite{russo2011calibrating} calibrated three different affine stochastic mortality models using term assurance premiums of three Italian insurance companies, and proposed that such affine models can be used for pricing mortality-linked securities. Thus, in the same spirit, we assume that $\lambda(t)$ evolves as
\eqlnostar{eq:dl}{
    \dd \lambda(t)=\left(a_{\lambda} (t)-b_{\lambda} \lambda(t)\right)\dd t +\sigma_{\lambda}\sqrt{\lambda(t)}\dd W_2 (t),&& \lambda(0)=\lambda_0,
}
where $a_\lambda(t)$ is a deterministic function, $b_\lambda$ and $\sigma_\lambda$ are positive constants. We restrict the mortality model parameters to satisfy the condition $2a_\lambda(t)>\sigma_\lambda^2 ,$ to ensure the strict positivity of $\lambda(t)$. The initial value of the mortality intensity $\lambda_0$ is calculated according to the Gompertz-Makeham law and is given by 
\eqlnostar{eq:lambda0}{
\lambda_0 = \phi +\frac{1}{b} e^{\frac{t_0 -m}{b}},
} 
where $t_0, m, \phi$ and $b$ are constants. As argued in \cite{menoncin2009death}, we wish that the expected value of $\lambda(t)$ equals to the Gompertz-Makeham force of mortality to ensure at any time $t$, $\lambda(t)$ has a reasonable value. To achieve this, we suppose that $a_{\lambda} (t)$ is of the following form 
\eqlnostar{eq:al}{
    a_\lambda (t)=b_\lambda \left(\phi +\left(\frac{1}{b_\lambda b}+1\right)\frac{1}{b}e^{\frac{t-m}{b}}\right).
}

The force of mortality is used as a tool to study the instantaneous survival rate of a population. If we denote by $p(t)$ the fraction of a population that survives from time $0$ to $t$, then $p(t)$ measures the cumulative survival rate which coincides with the survival probability. Since the force of mortality measures the instantaneous rate of mortality, we can write
\eqlnostar{eq:dp}{
\frac{\dd p(t)}{p(t)}=-\lambda(t)\dd t, && p(0)=1.
}
Given information up to time $t\in\Tc$, the (conditional) expected survival probability from $t$ to $s>t$ is given by (see \citet[Section 2.2]{menoncin2008role})
\eqstar{
\Eb\left[\frac{p(s)}{p(t)} \biggm| \Fc(t)\right]=\Eb \left[ e^{-\int_t^s\lambda(u)\dd u}\biggm| \Fc (t)\right].
} 

To discuss the prices of tradeable financial risky assets in the market, we first introduce a risk-neutral pricing measure $\Pbt$ by the following Radon-Nikodym derivative
\eqstar{
\frac{\dd \Pbt}{\dd \Pb}=Z(t)=\exp\left(-\int_{0}^{t}\Theta(s)^\prime \dd W(s)-\frac{1}{2}\int_{0}^{t}|\Theta(s)|^{2}\dd s \right).
}
In the above $\{\Theta(t)\mid t\in \Tc\}=\bigl\{\bigl[\theta_1(t),\theta_2(t),\theta_3(t)\bigr]^\prime \bigm| t\in \Tc\bigr\}$ is an $\Rb^3 $-valued, $\Fc$-adapted process such that $Z(t)$ is a martingale and $\Eb[Z(t)]=1$. Here, we denote by $\Eb[\cdot]$ the expectation operator under $\Pb$. By Girsanov's theorem, $\bigl\{\Wt (t) \mid t \in \Tc\bigr\}=\bigl\{\bigl[\Wt_1 (t),\Wt_2 (t),\Wt_3 (t)\bigr]^\prime \bigm| t \in \Tc\bigr\}$ is a three-dimensional standard Brownian motion under $\Pbt$ such that 
\eqlnostar{eq:Girsanov}{
    \Wt (t)=W(t)+\int_0 ^t \Theta(s)\dd s.
} 
The introduction of a risk-neutral measure also allows us to motivate the idea of market price of risk or risk-premium through $\Theta(t)$ in our financial market framework.

The first financial asset in the market is a representative stock. We suppose that the stock price process $S(t)$ under $\Pb$ evolves as
\eqstar{
    \frac{\dd S(t)}{S(t)}=&\left(r(t)+\theta_r\sigma_S^r r(t)+\theta_S\sigma_S\right)\dd t + \sigma_S^r \sqrt{r(t)}\dd W_1 (t)+\sigma_S \dd W_3 (t),&& S(0)=S_0,
} 
In the above, $\sigma_S, \sigma^r_S, \theta_r, \theta_S$ are some constants. Here, we have assumed that the market prices of interest rate risk and stock risk are $\theta_1(t)=\theta_r\sqrt{r(t)}$ and $\theta_3(t)=\theta_S$, respectively. 
The instantaneous covariance between the stock price and risk-free interest rate is captured by $\sigma_S^r\sqrt{r(t)}.$ The market price of stock risk and different volatility coefficients could be stochastic and take many different forms. However, as we mainly focus on the interest rate risk and longevity risk rather than investment risk, it is reasonable to suppose that they are constants.

For the pricing of a zero-coupon bond $B(t,T_B)$ which pays one unit of currency at a fixed maturity time $T_B,$ we first introduce a money market account $R(t)$: 
\eqstar{
\frac{\dd R(t)}{R (t)}=r(t)\dd t,& &R(0)=1. 
}
The risk-neutral pricing formula then give us 
\eqstar{
    B(t,T_B)=\Ebt \left[\frac{R(t)}{R(T_B)} \Biggm| \Fc(t)\right]=\Ebt\left[e^{-\int_t ^{T_B} r(u)\dd u} \biggm| \Fc (t) \right],
}
where $\Ebt[\cdot]$ is the expectation operator under measure $\Pbt$. As the interest rate $r(t)$ follows an affine model, we could solve for the bond price as
\eqlnostar{eq:affineB}{
B(t,T_B)=e^{f_0 (t,T_B)-f_1(t,T_B) r(t)}.
}
where
\eqlnostar{eq:bondf}{
f_0 (t,T_B)&=\frac{2a_r}{\sigma_r ^2}\log\left(\frac{2\eta_r e^{\frac{1}{2}(\tilde{b}_r + \eta_r)(T_B-t)}}{(\tilde{b}_r + \eta_r)(e^{\eta_r (T_B-t)}-1)+2\eta_r} \right),\\
f_1 (t,T_B)&=\frac{2(e^{\eta_r(T_B-t)} -1)}{(\tilde{b}_r + \eta_r )(e^{\eta_r(T_B-t)} -1)+2\eta_r},\\ \nonumber
\eta_r&=\sqrt{\tilde{b}_r ^2 +2\sigma_r ^2},\quad\quad \tilde{b}_r =b_r+\theta_r \sigma_r.
}
Such a formula can be found in several sources, for example, \citet[Section 3.2.3]{brigo2007interest}, \citet[Section 3.1.2]{cuchiero2006affine}. The dynamics of $B(t,T_B)$ under $\Pb$ is given as 
\eqstar{
    \frac{\dd B(t,T_B)}{B(t,T_B)}=\Big(r(t)+\theta_r\sqrt{r(t)}\sigma_B(t,T_B)\Big)\dd t+ \sigma_B(t,T_B)\dd W_1(t),
}
where we denote $\sigma_B(t,T_B)=-f_1(t,T_B)\sigma_r\sqrt{r(t)}.$ We include the complete calculations in Appendix \ref{sec:app1}.

As argued in \cite{boulier2001optimal}, it is convenient to use a single bond with a rolling maturity to replicate any bond on the market. Thus, we introduce a rolling bond $B(t)$ (with a little abuse of notation) with a constant time to maturity $T_B$. The price process of $B(t)$ is described by the following stochastic differential equation (SDE):
\eqlnostar{eq:rollingbond}{
    \frac{\dd B(t)}{B(t)}=\left(r(t)+\theta_r\sqrt{r(t)}\sigma_B(t,t+T_B)\right)\dd t+\sigma_B(t,t+T_B)\dd W_1(t).
}
Through the following equation, we can see that the zero-coupon bond $B(t,T_B)$ is replicable using cash and the rolling bond $B(t)$:
\eqstar{
    \frac{\dd B(t,T_B)}{B(t,T_B)}=\left(1-\frac{\sigma_B(t,T_B)}{\sigma_B(t,t+T_B)}\right)\frac{\dd R(t)}{R(t)}+\frac{\sigma_B(t,T_B)}{\sigma_B(t,t+T_B)}\frac{\dd B(t)}{B(t)}.
}
Thus, the use of rolling bond is equivalent to using a fixed maturity zero-coupon bond in the market. 

The third and final asset in the market is a zero-coupon longevity bond, which is primarily used to hedge the longevity risk. 
\begin{definition}
A zero-coupon longevity bond is a contract paying a face amount equal to the survival probability of the reference population from time 0 until a fixed maturity time.
\label{def:longevitybond}
\end{definition}
\noindent According to Definition \ref{def:longevitybond}, the payment of a zero-coupon longevity bond at a fixed maturity time $T_L$ is $p(T_L)$. Suppose that the market price of longevity risk is $\theta_2(t)=\theta_\lambda\sqrt{\lambda(t)}$, then the arbitrage-free price $L(t,T_L)$ of a zero-coupon longevity bond with fixed maturity time $T_L$ is given as 
\eqstar{
    L(t,T_L)&=\Ebt\left[\frac{R(t)}{R(T_L)}p(T_L) \biggm| \Fc(t)\right]=e^{-\int_0 ^t \lambda(u)\dd u}\Ebt\left[e^{-\int_t ^{T_L} \left(r(u)+\lambda(u)\right)\dd u} \biggm| \Fc(t)\right].
}
Due to the affine nature of $r(t)$ and $\lambda(t)$ and the independence between them, the longevity bond price can be expressed in the following form
\eqlnostar{eq:affineL}{
L(t,T_L ) =&e^{-\int_0 ^t \lambda(u)\dd u}N(t,T_L),
}    
where 
\eqlnostar{eq:bondh}{
N(t,T_L) &= e^{f_0 (t,T_L )-f_1 (t,T_L )r(t)+h_0 (t,T_L )-h_1 (t,T_L )\lambda(t)},\\
h_1 (t,T_L) &=\frac{2(e^{\eta_\lambda(T_L-t)} -1)}{(\tilde{b}_\lambda + \eta_\lambda )(e^{\eta_\lambda(T_L-t)} -1)+2\eta_\lambda},\\
h_0 (t,T_L) &=-\int_t^{T_L} a_\lambda(u)h_1(u,T_L) \dd u,\\
\eta_\lambda &=\sqrt{\tilde{b}_\lambda ^2 +2\sigma_\lambda ^2}, \quad\quad \tilde{b_\lambda} =b_\lambda+\theta_\lambda\sigma_\lambda.
}
By denoting $\sigma_L^r (t,T_L)=-f_1 (t,T_L)\sigma_r \sqrt{r(t)}$ and $\sigma_L^\lambda(t,T_L)=-h_1(t,T_L) \sigma_\lambda\sqrt{\lambda(t)}$, the evolution of $L(t,T_L)$ is then described as
\eqstar{
    \frac{\dd L(t,T_L )}{L(t,T_L )}=&\Big(r(t)+\theta_r \sqrt{r(t)}\sigma_L^r (t,T_L ) +\theta_\lambda \sqrt{\lambda(t)}\sigma_L^\lambda (t,T_L)\Big)\dd t\\
    &+\sigma_L^r (t,T_L )\dd W_1 (t)+\sigma_L ^\lambda(t,T_L)\dd W_2 (t).
} 
See Appendix \ref{sec:appendix2} for detailed calculations. 

In the same manner as zero-coupon bond, we consider a rolling longevity bond $L(t)$ (with a little abuse of notation) with a constant time to maturity $T_L$ whose price process  under $\Pb$ is given as:
\eqlnostar{eq:rollingL}{
    \frac{\dd L(t)}{L(t)}=&\left(r(t)+\theta_r\sqrt{r(t)}\sigma_L^r(t,t+T_L)+\theta_\lambda\sqrt{\lambda(t)}\sigma_L^\lambda(t,t+T_L)\right)\dd t\\
    &+\sigma_L^r(t,t+T_L)\dd W_1(t)+\sigma_L^\lambda(t,t+T_L)\dd W_2(t).
}
We see that the rolling longevity bond correlates with interest rate $r(t)$ as well as force of mortality $\lambda(t)$. In fact, zero-coupon longevity bonds with any fixed maturity can be replicated using rolling bond, rolling longevity bond and cash:
\eqstar{
    \frac{\dd L(t,T_L)}{L(t,T_L)}=n_0(t)\frac{\dd R(t)}{R(t)}+n_B(t)\frac{\dd B(t)}{B(t)}+n_L(t)\frac{\dd L(t)}{L(t)},
}
where
\eqstar{
n_L(t)=\frac{\sigma_L^\lambda(t,T_L)}{\sigma_L^\lambda(t,t+T_L)}, & & n_B(t)=\frac{\sigma_L^r(t,T_L)}{\sigma_B(t,t+T_B)}-n_L(t)\frac{\sigma_L^r(t,t+T_L)}{\sigma_B(t,t+T_B)}, & & n_0(t)=1-n_B(t)-n_L(t).
}
It is common to use rolling bonds in the literature: \cite{han2012optimal} introduced a rolling indexed bond to hedge the inflation risk for a DC scheme. \cite{menoncin2008role} used rolling longevity bond to transfer an individual's longevity risk. Indeed, the proposed problem in this work can also be solved using a fixed maturity zero-coupon longevity bond and zero-coupon bond. The use of rolling longevity bond and rolling bond only simplifies the calculations in Section \ref{sec:main}. Moreover, our specific choice of market prices of risks $\theta_r\sqrt{r(t)}$ and $\theta_\lambda\sqrt{\lambda(t)}$ maintains the affine form of our models (see, for example, \cite{duffee2002term}).

For any $t,T_B,T_L\in \Tc$, we describe the risky asset prices in the form of a vector: 
\eqlnostar{eq:market}{
    \left[\begin{array}{ccc}
    \frac{\dd B(t)}{B(t)}\\
    \frac{\dd L(t)}{L(t)}\\
    \frac{\dd S(t)}{S(t)}
    \end{array}\right]
    =\left(r(t)\mathbb{1}+M(t)\right)\dd t+\Sigma(t)^\prime \dd W(t),
}
where 
\eqstar{
    &M(t)=\left[\begin{array}{ccc}
    \theta_r \sqrt{r(t)}\sigma_B (t,t+T_B)\\
    \theta_r \sqrt{r(t)}\sigma_L^r (t,t+T_L)+\theta_\lambda \sqrt{\lambda(t)}\sigma_L^\lambda (t,t+T_L)\\
    \theta_r\sigma_S^r  r(t)+\theta_S\sigma_S 
    \end{array}\right],\\
    &\Sigma (t)^\prime=\left[\begin{array}{ccc}
    \sigma_B (t,t+T_B) & 0 & 0\\
    \sigma_L^r (t,t+T_L) & \sigma_L^\lambda (t,t+T_L) & 0 \\
    \sigma_S^r\sqrt{r(t)} & 0 & \sigma_S 
    \end{array}\right].
}
For ease of presentation, we also denote by $z(t)=[r(t), \lambda(t)]^\prime$ whose dynamics is given as
\eqlnostar{eq:state}{
\dd z(t)=\mu (t,z(t))\dd t+\xi(t,z(t))^\prime \dd W(t), &&  z(0) = [r_0, \lambda_0]^\prime,
}
where
\eqstar{
    \mu (t,z(t))=\left[\begin{array}{ccc}
    a_r -b_r r(t)\\
    a_\lambda (t)-b_\lambda\lambda(t)
    \end{array} \right], & &\xi (t,z(t))^\prime=\left[\begin{array}{ccc}
    \sigma_r\sqrt{r(t)} & 0 & 0\\
    0 & \sigma_\lambda \sqrt{\lambda(t)}  & 0
    \end{array}\right].
}

\section{Main results}
\label{sec:main}
In the literature, a representative member has been used to study the optimal asset allocation problem for DC schemes, for instance, see \cite{boulier2001optimal}. In this work, we also consider a representative member who continuously contributes a fraction of his wage into the pension scheme during the accumulation phase. At retirement time $T\in \Tc$, the accumulated contributions are used to purchase a lifetime annuity to provide regular incomes in retirement. The scheme manager works on behalf of the representative member and decides the investment strategy in the accumulation phase to increase the scheme wealth. A minimum guarantee is introduced in the scheme to protect the representative member against poor investment performance. There are two main types of optimality criterion used in the literature for portfolio selection problems: the utility maximisation criterion and the mean-variance criterion. In our framework, we follow \cite{merton1969lifetime} and assume that the scheme manager aims to maximise the representative member's expected utility from the surplus between the terminal wealth and the minimum guarantee.

\subsection{The utility maximisation problem}
\label{sec:criterion}
We assume that the representative member contributes a fraction $r_c$ of his instantaneous wage $w(t).$ Previous studies, such as \cite{han2012optimal} and \cite{guan2014optimal}, model the wage (or, contribution) as a stochastic process to study the optimal asset allocation problem for DC schemes. To simplify our calculations, the instantaneous wage in this work is assumed to be constant, that is, for any $t\in [0,T]$, $w(t)=w.$ Thus, the contribution $c(t) = r_c w(t) = r_c w = c$ is also constant. We note that our following analysis is also applicable when $w(t)$ and $c(t)$ are treated as independent stochastic processes or deterministic functions. During the accumulation phase, at any time $t\in[0,T]$, the scheme manager invests $ \alpha_S (t), \alpha_B (t)$ and $\alpha_L (t)$ amounts of money in stock, rolling bond and rolling longevity bond, respectively. It is clear that the amount of money invested in money market is $\alpha_0(t) = F(t) - \alpha_B(t) - \alpha_L(t) -\alpha_S(t),$ where $F(t)$ denotes the scheme's wealth level. The death time of the representative member $\tau$ is assumed to be a random variable on $(\Omega,\Fc,\Pb)$. In the case where the member dies before the retirement time, that is, $0<\tau<T$, we assume that his heirs receive his total pension as bequest. The dynamics of $F(t)$ is given as
\eqlnostar{eq:fundwealth}{
\dd F(t)=&\Big(r(t)F(t)+c+\alpha(t)^\prime M(t)\Big)\dd t+\alpha(t)^\prime \Sigma(t)^\prime \dd W(t), && F(0) = F_0,
}
where $\big\{\alpha (t)$ $\bigm| t\in [0,T]\big\}$ $=\big\{\left[\alpha_B (t),\alpha_L (t), \alpha_S (t)\right]^\prime \bigm| t\in [0,T]\big\}$ denote the investment in risky assets. Note here that the force of mortality $\lambda(t)$ is not involved in the scheme's wealth process since we assume that the heirs take the passed away members' pension (scheme's wealth) upon death. If the heirs were to receive only a fraction of the scheme's wealth upon the member's death, the wealth process will be influenced by $\lambda(t)$. We provide detailed calculations in Appendix \ref{appendix:wealth} which clarify this subtlety. 

The representative member uses his total pension wealth to purchase a lifetime annuity at retirement time $T$ and requires that the pension wealth must exceed this annuity price which acts as a minimum guarantee. The minimum guarantee was also previously considered in works such as \cite{boulier2001optimal}, \cite{deelstra2003optimal}, \cite{han2012optimal} and \cite{guan2014optimal}. We extend these works to the case where the death time is uncertain and the force of mortality is stochastic. To compute the price of the lifetime annuity, we first need to decide the level of instalments that the annuity delivers. Typically, the wage replacement ratio $r_w$, the percentage of retirement income to pre-retirement income, is a good estimate of the income needed to maintain the living standard in retirement. We set the instantaneous instalment of the annuity to be $\pi=r_w w$, so that the lifetime annuity provides sufficient retirement income for subsistence. By denoting $a(T)$ as the price of lifetime annuity at retirement time $T$, we have 
\eqstar{
a(T)=\Ebt \left[ \int_T^\infty \pi  \frac{R(T)}{R(s)}\frac{p(s)}{p(T)}\dd s \biggm| \Fc(T)\right].
}  
The minimum guarantee $G$ is to purchase life time annuity for the surviving members at retirement time $T$. Its value at $T$ is thus given as
\eqstar{
    G(T)&= p(T)a(T)=\Ebt \left[\int_T ^\infty \pi \frac{R(T)}{R(s)} p(s)\dd s \biggm| \Fc(T)\right].
} 
As mentioned earlier, the goal of the pension scheme manager is to maximise the expected utility from the surplus between the fund level and the minimum guarantee at retirement time $T$. Thus, for a given investment strategy $\alpha,$ the manager's objective function for the utility maximisation problem is
\eqstar{J(t,f,z;\alpha)=\Eb\left[U\Big(F(T)-G(T)\Big)\right], \quad 0< t \leq T,
} 
where $U:\Rb_+ \to \Rb_+$ is a utility function, $F(t) =f >0$ and $z(t) = z$. 

In the literature, many different utility functions are considered in the problems related to optimal investment. Among the several choices, hyperbolic absolute risk aversion (HARA) class of utility functions is the most commonly used. Constant relative risk aversion (CRRA), constant absolute risk aversion (CARA), and quadratic utility are all types of HARA functions that have been used in previous works. For instance, \cite{gao2008stochastic} and \cite{boulier2001optimal} used CRRA utility function to study the optimal asset allocation problems for DC schemes. \cite{battocchio2004optimal} and \cite{cairns2000some} adopted CARA and quadratic utilities, respectively. In this paper, we use the power utility function (CRRA):
\eqlnostar{eq:utility}{
U(x)=\frac{x^{1-\gamma}}{1-\gamma},
}
where $\gamma>0$ and $\gamma\not=1$. In the case when $\gamma =1$, $U(x)=\ln x$ is the log-utility function. 

Our choice of the power utility function is motivated by two reasons. First, pension schemes usually manage a large amount of money. With an increasing or decreasing relative risk aversion, the fraction of the wealth invested in risky assets is affected by the level of wealth. However, the power utility function has a constant relative risk aversion and the investment strategy is not affected by scale. Second, our optimisation problem is analytically tractable when using the power utility function. We lose the analytical tractability for other types of utility functions even if we can numerically solve the optimisation problem using our approach. 

Thus, we formulate the scheme manager's utility maximisation problem as
\eqlnostar{eq:problem1}{
\underset{\alpha \in\Ac}{\sup} \, \Eb \left[\frac{(F(T)-G(T))^{1-\gamma}}{1-\gamma} \right] \text{ such that } F(T) \geq G(T) \, \text{a.s.}.
}
In the above, the set $\Ac$ denotes the set of all \textit{admissible} strategies which are defined as below.
\begin{definition}
A portfolio strategy $\{\alpha (t)  \in \Rb^3\mid t \in [0,T] \}$ is called admissible if $\alpha(t)$ is progressively measurable with respect to $\Fc$ and $\Eb \left[\int_0^T |\alpha(t)|^2 \dd t\right] < \infty$.
\end{definition} 
In our setting, the representative member continuously contributes to the scheme during the accumulation phase. The $c\dd t$ term in \eqref{eq:fundwealth} reveals that the wealth process $F(t)$ is not self-financing. Besides, at the retirement time $T$, there is a minimum guarantee $G(T)$ to be  met. This means that the proposed problem (\ref{eq:problem1}) is not a classical Merton's optimal investment problem. To solve this non-self-financing constrained problem, we transfer it into a self-financing investment portfolio optimisation problem by introducing an auxiliary surplus process. We then solve the transformed problem using the dynamic programming principle.

\subsection{Single investment portfolio optimisation problem}
\label{sec:uncostrained}
Inspired by \cite{boulier2001optimal}, we split the scheme's wealth into two parts: one part is the future contributions to be paid and the other part is a self-financing portfolio. For any $t\in[0,T]$, denoting by $D(t)$ the present value of future contributions until retirement time $T$, we can write
\eqlnostar{eq:future contributions}{
D(t) =\Ebt\left[\int_t^T c\frac{R(t)}{R(s)}\dd s\Biggm| \Fc(t)\right].
}
$D(t)$ can be viewed as a coupon-paying bond that pays a continuous coupon at rate $c$ until time $T$. Thus, $D(t)$ can be replicated by investing in the rolling bond and the money market.
\begin{proposition}
\label{prop:D}
For any $t\in [0,T]$, $D(t)$ in \eqref{eq:future contributions} can be replicated as
\eqlnostar{eq:contribution}{
    \dd D(t)=-c\dd t+\alpha_0^B(t)\frac{\dd R (t)}{R (t)}+\alpha_B ^D (t)\frac{\dd B(t)}{B(t)},
} where 
\eqlnostar{strategyD}{
\alpha_B ^D (t)=\frac{ c\int_t ^TB(t,s)f_1 (t,s)\dd s}{f_1 (t,t+T_B )},& &\alpha_0^B(t) = D (t)-\alpha_B^D(t),
} 
are the holdings in rolling bond and money market, respectively. 
\end{proposition} 
See Appendix \ref{appendix:dD} for the proof. 

Our next step is to transform the problem \eqref{eq:problem1} to a simplified portfolio optimisation problem with the following steps. First, we construct a replicating portfolio for $G(T)$. At time $t\in[0,T]$, the present value of $G(T)$ is given by 
\eqlnostar{eq:definition guarantee}{
G(t) =& \widetilde{\Eb} \left[G(T)\frac{R(t)}{R(T)}\Biggm| \Fc(t)\right].
}
In Proposition \ref{prop:D}, we replicate $D(t)$ by using cash and rolling bond since the interest rate risk is the only risk to hedge. We show in the following proposition that $G(t)$ can be similarly replicated by investing in bond, longevity bond and money market.
\begin{proposition}
\label{prop:G}
For any $t\in[0,T]$, $G(t)$ in \eqref{eq:definition guarantee} can be replicated as
\eqlnostar{eq:guarantee}{
\dd G(t)=&\alpha_0^G (t)\frac{\dd R(t)}{R(t)}+\alpha_B^G (t)\frac{\dd B(t)}{B(t)}+\alpha_L^G (t)\frac{\dd L(t)}{L(t)},
}
where
\eqlnostar{strategyG}{
\alpha_L^G (t)=&\frac{\pi  \int_T^\infty L(t,s)h_1(t,s)\dd s}{h_1(t,t+T_L)},\\
\alpha_B^G(t)=&\frac{\pi  \int_T^\infty L(t,s)f_1(t,s)\dd s}{f_1(t,t+T_B)}-\alpha_L^G (t)\frac{f_1(t,t+T_L)}{f_1(t,t+T_B)},\\
\alpha_0^G(t) = &G(t)-\alpha_B^G(t)-\alpha_L^G(t).
} 
are the holdings in rolling longevity bond, rolling bond and money market, respectively.
\end{proposition}
See Appendix \ref{appendix:dG} for the proof.

Finally, we construct an auxiliary process $Y(t)$ which is the surplus of the terminal wealth over the minimum guarantee,  that is, $Y(t)=F(t)+D(t)-G(t)$. At retirement time $T$, from \eqref{eq:future contributions}, we see that $D(T)=0$ and we have $Y(T)=F(T)-G(T)$. From \eqref{eq:fundwealth}, \eqref{eq:contribution} and \eqref{eq:guarantee}, we obtain the following equation
\eqstar{
    \dd Y(t)=&\dd F(t)+\dd D(t)-\dd G(t)\\
    =&\alpha_0^Y(t)\frac{\dd R(t)}{R(t)}+\alpha_B^Y (t)\frac{\dd B(t)}{B(t)}+\alpha_L^Y (t)\frac{\dd L(t)}{L(t)}+\alpha_S^Y(t)\frac{\dd S(t)}{S(t)},
}
where
\eqlnostar{eq:strategyY}{
\alpha_B^Y (t)=\alpha_B (t)+\alpha_B^D (t)-\alpha_B^G (t), && \alpha_L^Y (t)=\alpha_L (t)-\alpha_L^G(t),\nonumber\\
\alpha_S^Y (t)=\alpha_S (t), && \alpha_0^Y(t)=Y(t)-\alpha_B^Y(t)-\alpha_L^Y (t)-\alpha_S^Y (t).
}
For any $t\in[0,T]$, let 
\eqstar{
\alpha^D (t)=[\alpha_B^D (t),0,0]^\prime, && \alpha^G (t)=[\alpha_B^G (t),\alpha_L^G (t),0]^\prime, && \alpha^Y (t)=[\alpha_B^Y (t), \alpha_L^Y (t),\alpha_S^Y (t)]^\prime.
}
Then, we have 
\eqlnostar{strategy}{
\alpha^Y(t)=\alpha(t)+\alpha^D(t)-\alpha^G(t).
}
The dynamics of the surplus process $Y(t)$ can be then written as
\eqlnostar{eq:Y}{
\dd Y(t)=\left(r(t)Y(t)+\alpha^Y (t)^\prime M(t)\right)\dd t+\alpha^Y (t)^\prime \Sigma(t)^\prime \dd W(t).
}

Thus, our simplified portfolio optimisation problem is given as
\eqlnostar{eq:problem3}{
     &\underset{\alpha^Y  \in\Ac}{\sup} \, \Eb\left[\frac{Y(T)^{1-\gamma}}{1-\gamma} \right] \text{ such that } Y(T) \geq 0 \, \text{a.s.}.
} 

\begin{lemma} 
\label{remark:equivalent}
For any $t\in[0,T]$, if $\alpha^Y(t)\in \Ac$, then $\alpha(t)\in\Ac,$ and the optimisation problems \eqref{eq:problem1} and \eqref{eq:problem3} are equivalent.
\end{lemma}
\begin{proof}
It is clear from \eqref{strategy} that we need to prove the admissibility condition for deterministic functions $\alpha^D(t)$ and $\alpha^G(t).$ Now for any fixed $t\in[0,T]$ and any $s\in[T,\infty)$, $f_1(t,s)$ and $B(t,s)$ are continuous functions. It is easy to see that $\Eb\left[\int_0^T |\alpha^D(t)|^2\dd t\right]<+\infty$. $h_1(t,s)$ and $L(t,s)$ are also continuous functions. Moreover, as $s\to \infty$, functions $L(t,s)f_1(t,s) \to 0,$ and $L(t,s)h_1(t,s) \to 0.$ Then for any $\epsilon>0$, there exists a $K\ge T>0$ such that for all $ K_2>K_1\geq K$, we have 
\eqstar{
\Big|\int_{K_1}^{K_2} L(t,s)f_1(t,s)\dd s \Big|<\epsilon,&&\Big|\int_{K_1}^{K_2} L(t,s)h_1(t,s)\dd s \Big|<\epsilon.}
According to the Cauchy criterion, $\int_T^\infty L(t,s)f_1(t,s)\dd s$, $\int_T^\infty L(t,s)h_1(t,s)\dd s$ are then convergent. Thus, $\Eb\left[\int_0^T |\alpha^G(t)|^2\dd t\right]<+\infty$. Therefore, if $\Eb\left[\int_0^T |\alpha^Y(t)|^2 \dd t\right]<+\infty$, then $\Eb\left[\int_0^T |\alpha(t) |^2\dd t\right]<+\infty$. This means that if $\alpha^Y(t)\in \Ac$, then $\alpha(t)\in \Ac$. We have $D(T)=0$, thus if $Y(T)\ge 0$ a.s., then $F(T)-G(T)\geq 0$ a.s.. Furthermore, since $F(t)=Y(t)-D(t)+G(t)$ and \eqref{strategy} holds, ${\alpha^Y}^\star(t)$ leads to the optimal strategy $\alpha^\star(t)$ which concludes the argument.
\end{proof}
We see that the wealth process \eqref{eq:Y} is self-financing. Once we are able to solve \eqref{eq:problem3} and obtain the unique optimal control ${\alpha^Y}^\star(t)$, we can use \eqref{strategyD}, \eqref{strategyG} and \eqref{strategy} to obtain $\alpha^\star(t)$.

\subsection{The optimal solution}
\label{sec:solution}
We define the value function of our simplified problem
\eqref{eq:problem3} as 
\eqstar{V(t,y,z):=\underset{\alpha^Y \in\Ac}{\sup} \,\Eb\left[\frac{Y(T)^{1-\gamma}}{1-\gamma}\right],
} 
with terminal condition $V(T,y,z)=\frac{Y(T)^{1-\gamma}}{1-\gamma}$. We assume that the value function $V(t,y,z) \in C^{1,2,2}([0,$ $T] \times \Rb^3_+).$ Then, by following the usual dynamic programming principle (see, for example, \citet[Chapter 3]{pham2009continuous}), $V$ satisfies the following Hamilton-Jacobi-Bellman (HJB) equation 
\eqlnostar{eq:hjb1}{
0=V_t(t,y,z)+\underset{\alpha^Y \in \Rb^3}{\sup}\Lc^{\alpha^Y}V(t,y,z)
}
where 
\eqstar{
\Lc^{\alpha^Y}V = \left[V_y (ry+{\alpha^Y}^\prime M)+\mu^\prime V_z+\frac{1}{2}\text{tr}(\xi^\prime \xi V_{zz} )+\frac{1}{2}{\alpha^Y}^\prime\Sigma^\prime \Sigma\alpha^Y V_{yy}+{\alpha^Y}^\prime\Sigma^\prime \xi V_{yz}\right].
} 
$V_t$, $V_y,\ V_{yy},\ V_z,\ V_{zz}\ \text{and}\ V_{yz}$ are the first and second order partial derivatives with respect to $t, y, z$. In particular, $V_z=\left[V_r, V_\lambda\right]^\prime$, $V_{yz}=\left[V_{yr},V_{y\lambda} \right]^\prime$ and $V_{zz}=\left[ \left[V_{rr},V_{\lambda r}\right]^\prime,\left[V_{r\lambda }, V_{\lambda\lambda}\right]^\prime \right]$. Solving the first order condition on $\alpha^Y$, we obtain the unique investment strategy as
\eqlnostar{eq:fod}{
{\alpha^Y}^\star=-\frac{V_y}{V_{yy}} (\Sigma^\prime \Sigma)^{-1} M-\frac{1}{V_{yy}} \Sigma^{-1}  \xi V_{yz}.
}            
Substituting \eqref{eq:fod} in \eqref{eq:hjb1}, we obtain 
\eqlnostar{eq:hjb}{
    0=&V_t+V_y ry-\frac{1}{2}\frac{{V_y}^2}{V_{yy}} M^\prime (\Sigma^\prime \Sigma)^{-1} M
    -\frac{V_y}{V_{yy}} M^\prime \Sigma^{-1} \xi V_{yz}+\mu^\prime V_z\\
    &+\frac{1}{2}tr(\xi^\prime \xi V_{zz})-\frac{1}{2} \frac{1}{V_{yy}} {V_{yz}}^\prime \xi^\prime \xi V_{yz}.
}
Once we solve the value function $V(t,y,z)$ in \eqref{eq:hjb}, we can obtain the optimal control ${\alpha^Y}^\star(t)$. The following proposition gives the explicit optimal investment strategy for the transformed problem \eqref{eq:problem3}.

\begin{proposition}
\label{prop:problem3}
For any $t\in[0,T]$, with a risk-aversion parameter 
\eqlnostar{eq:gamma condition}{
\gamma>\max\left\{\frac{2\sigma_r^2+\sigma_r^2 \theta_r^2+2b_r \theta_r \sigma_r}{(b_r+\theta_r\sigma_r)^2+2\sigma_r^2},\ \frac{2b_\lambda \theta_\lambda \sigma_\lambda +\sigma_\lambda^2 \theta_\lambda^2}{(b_\lambda+\theta_\lambda\sigma_\lambda)^2}\right\}
}  and under the financial market setting \eqref{eq:market}--\eqref{eq:state}, the optimal solution to \eqref{eq:problem3} is given as 
\eqstar{
{\alpha^Y_B}^\star(t)=&\frac{\theta_S\sigma_S^r -\theta_r\sigma_S-\sigma_S\sigma_rA_1(t,T)}{\gamma\sigma_S\sigma_r f_1 (t,t+T_B )}Y(t)-\frac{f_1 (t,t+T_L ) }{f_1 (t,t+T_B )}{\alpha^Y_L}^\star(t),\\
{\alpha^Y_L}^\star(t)=&-\frac{\theta_\lambda+\sigma_\lambda A_2(t,T)}{\gamma\sigma_\lambda h_1 (t,t+T_L)}Y(t),\\
{\alpha^Y_S}^\star(t)=&\frac{\theta_S}{\gamma\sigma_S }Y(t),\quad{\alpha^Y_0}^\star(t)=Y(t)-{\alpha^Y_B}^\star(t)-{\alpha^Y_L}^\star(t)-{\alpha^Y_S}^\star(t),
}
where 
\eqlnostar{eq:A}{
    &\left\{\begin{array}{ll}
         &A_1 (t,T)=\frac{a_{11} a_{12}  \exp(-\sqrt{\Delta_1}(T-t)) -a_{11} a_{12}}{a_{12} \exp(-\sqrt{\Delta_1}(T-t))-a_{11}},\\
         &\Delta_1 =b_r^2+\frac{\gamma-1}{\gamma}\left(2\sigma_r^2 +\theta_r^2\sigma_r^2+2b_r\theta_r\sigma_r\right),\\
         &a_{11,12}=\frac{(\gamma-1)\theta_r\sigma_r +b_r \gamma\pm \gamma\sqrt{\Delta_1}}{\sigma_r^2},
    \end{array}\right.\nonumber\\
    &\left\{\begin{array}{ll}
         &A_2 (t,T)=\frac{a_{21} a_{22} \exp(-\sqrt{\Delta_2}(T-t)) -a_{21} a_{22}}{a_{22} \exp(-\sqrt{\Delta_2}(T-t))-a_{21}},\\
         &\Delta_2=b_\lambda^2+\frac{\gamma-1}{\gamma}(2b_\lambda\theta_\lambda\sigma_\lambda+\theta_\lambda^2\sigma_\lambda^2),\\
         &a_{21,22}=\frac{(\gamma-1)\theta_\lambda \sigma_\lambda+b_\lambda \gamma\pm \gamma\sqrt{\Delta_2}}{\sigma_\lambda^2}, 
    \end{array}\right.\nonumber\\
    &A_0(t,T)=\int_t^T \left(a_r A_1 (s,T)+a_\lambda(s) A_2(s,T)+\frac{1-\gamma}{2\gamma}\theta_S^2\right)\dd s.
}
\label{Prop:solution1}
\end{proposition}
We include the proof in Appendix \ref{sec:appendixc}. 

Next, we verify the admissibility of the optimal control ${\alpha^Y}^\star(t)$.
\begin{lemma}
\label{lemma:verify}
For any $t\in[0,T]$, the optimal strategy ${\alpha^Y}^\star(t)$ in Proposition \ref{Prop:solution1} belongs to the admissible set $\Ac$.
\end{lemma}

\begin{proof}
For any $t\in[0,T]$, let $\tilde{Y}(t)=\ln Y(t)$, we have
\eqstar{
\dd \tilde{Y}(t) = \left(r(t)+ \tilde{\alpha}^Y (t)^\prime  M(t)\right)\dd t+\tilde{\alpha}^Y (t)^\prime \Sigma^\prime (t)\dd W(t),
}
where $\tilde{\alpha}^Y (t)=\frac{\alpha^{Y}(t)}{Y(t)}$. According to Proposition \ref{Prop:solution1}, it is easy to see that the optimal $\tilde{\alpha}^{Y^\star}(t)$ is a vector of continuous functions. Thus, $\Eb\left[\int_0^T|\tilde{\alpha}^{Y^\star}(t) |^2 \dd t\right]<+\infty$, $\tilde{Y}(t)$ admits a unique solution and is bounded over $[0,T]$. Furthermore, \eqref{eq:Y} has a unique solution $Y^\star(t)$ (obtained using $\alpha^{Y^\star}(t)$) which is square integrable. Thereafter, we have $\Eb\left[\int_0^T|\alpha^{Y^\star}(t) |^2 \dd t\right]<+\infty$ and ${\alpha^Y}^\star (t) \in \Ac$.
\end{proof}

From the above proposition, we observe a proportional relationship between ${\alpha^Y}^\star(t)$ and $Y(t)$ for constant $\theta_S$, $\sigma_S$ and $\gamma$. Namely, the optimal stock weight $\frac{{\alpha^Y_S}^\star(t)}{Y(t)}$ always stays the same throughout the investment horizon. This is similar to the classical Merton's portfolio problem where the optimal portfolio weight on the risky asset is constant over time. The convention is that the constant market price of risk causes no change in the investment behaviour. We also find that the optimal investment in longevity bond ${\alpha_L^Y}^\star(t)$ is actually taken from the investment in bond ${\alpha_B^Y}^\star(t)$ proportionally by a factor of $-\frac{f_1 (t,t+T_L )}{f_1 (t,t+T_B )}$. From \eqref{eq:affineB} and \eqref{eq:affineL}, we see that $f_1 (t,t+T_B )$ and $f_1 (t,t+T_L )$ are the duration of the rolling bond and rolling longevity bond, respectively. Since the duration is always positive, $-\frac{f_1 (t,t+T_L )}{f_1 (t,t+T_B )}$ is negative. We conclude that there is a negative relationship between the optimal investments in bond and longevity bond. Besides, if the maturities of the rolling bond and the rolling longevity bond are the same (that is, $T_B=T_L$), we have $f_1 (t,t+T_L )=f_1 (t,t+T_B )$ and ${\alpha_L^Y}^\star(t)$ is fully deduced from ${\alpha_B^Y}^\star(t)$. $h_1(t,t+T_L)$ is a  constant, $A_1(t,T)$ and $A_2(t,T)$ are positive and increase with $t$, thus $\frac{{\alpha_L^Y}^\star(t)}{Y(t)}$ is decreasing over time. Thus, the optimal investment proportion in longevity bond decreases with time. However, the behaviour of optimal proportions invested in bond and money market are not clear. Also, it is easy to deduce that greater the manager's risk-aversion, lower the portfolio weights on the longevity bond and stock. 

Lemma \ref{remark:equivalent} shows that Problem \eqref{eq:problem1} and \eqref{eq:problem3} are equivalent. According to \eqref{strategyD}, \eqref{strategyG}, \eqref{strategy} and Proposition \ref{Prop:solution1}, we obtain the solution to the initial problem \eqref{eq:problem1} and state it in the following proposition.
\begin{proposition}
\label{Prop:solution}
For any $t\in[0,T]$, under condition \eqref{eq:gamma condition} and the financial market setting \eqref{eq:market}--\eqref{eq:state}, the optimal solution to \eqref{eq:problem1} is given as 
\eqstar{
\alpha_B^\star (t)=&\frac{\theta_S\sigma_S^r -\theta_r\sigma_S-\sigma_S\sigma_rA_1(t,T)}{\gamma\sigma_S\sigma_r f_1 (t,t+T_B )}Y(t)+\frac{\pi  \int_T^\infty L(t,s)f_1(t,s)\dd s}{f_1(t,t+T_B)}\\
&-\frac{ \int_t ^TB(t,s)f_1 (t,s)\dd s}{f_1 (t,t+T_B )}-\frac{f_1 (t,t+T_L ) }{f_1 (t,t+T_B)}\alpha_L^\star (t),\\
\alpha_L^\star (t)=&-\frac{\theta_\lambda+\sigma_\lambda A_2(t,T)}{\gamma\sigma_\lambda h_1 (t,t+T_L)}Y(t)+\frac{\pi  \int_T^\infty L(t,s)h_1(t,s)\dd s}{h_1(t,t+T_L)},\\
\alpha_S^\star (t)=&\frac{\theta_S}{\gamma\sigma_S }Y(t),\quad \alpha_0 ^\star(t)=F(t)-\alpha_B^\star (t)-\alpha_L^\star (t)-\alpha_S^\star (t).
} 
\end{proposition}
The result is obtained by straightforward calculations. We find that the optimal amount invested in stock does not depend on the fund wealth $F(t)$ but on the process $Y(t)$. Besides, the optimal stock weight $\frac{\alpha_S^\star(t)}{F(t)}$ does not keep constant any more and depends on the ratio $\frac{Y(t)}{F(t)}$. From the drift terms of $F(t)$ and $Y(t)$ in \eqref{eq:fundwealth} and \eqref{eq:Y}, we deduce that $F(t)$ is expected to increase faster than $Y(t)$. Thus, we conclude that the optimal stock weight falls over time. Similar to the result in Proposition \ref{Prop:solution1}, the optimal investments in bond and longevity bond correlate negatively. We can not easily infer from the solution how the optimal weights in bond and longevity bond change over time. Section \ref{sec:numerics} shows the results of numerical simulations which allows us to observe the optimal investment strategy dynamically. It is not straightforward to detect, from the optimal solution, how the risk-aversion coefficient $\gamma$, contribution rate $r_c$ and wage replacement ratio $r_w$ affect the optimal strategy. We provide numerical analyses on these parameters in the following section.

\section{Numerical applications}
\label{sec:numerics}
We first give a base scenario to show the optimal proportions invested in risky assets and money market. Then, sensitivity analyses are provided to study the impact of model parameters on the optimal investment strategy. In what follows, we denote by $w_B(t)$, $w_L(t)$, $w_S(t)$ and $w_0(t)$ the investment proportions in bond, longevity bond, stock and money market, respectively.

\subsection{The base scenario}
The values of the parameters for the base scenario are given in Table \ref{tab:data}. We do not use real market data but most of these parameter values are taken from \cite{menoncin2017longevity} and \cite{han2012optimal}. Here, we assume that there exists a rolling bond and a rolling longevity bond with constant maturity time (in years) $T_B=10$ and $T_L=10$, respectively. As the longevity bond is supposed to be issued based on the mortality index of an older population, we assume that the market consists of a rolling longevity bond whose underlying is the survival index of the 40-year-old population. The longevity risk tends to be largely ignored in very early ages. Hence, we suppose that the  scheme manager considers to add the longevity bond to the investment portfolio after the representative scheme member reaches the age of 40, that is, the initial age (in years) is $t_0=40$. $\lambda_0$ given in \eqref{eq:lambda0} is computed by using the parameters given in Table \ref{tab:data}. The retirement age is chosen as 65 years old, in other words the retirement time is $T=25$ years. $\theta_\lambda$ is the parameter of the market price of longevity risk. It is hard to decide the value of $\theta_\lambda$ since there are no longevity bonds liquidly traded in the market. In our base scenario, we use $\theta_\lambda=-0.10$. At time 0, the rolling  bond offers a risk premium of about $0.01370$ and the longevity bond provides a total risk premium of around $0.01372$. The stock gives a total risk premium of around $0.01670$. In this setting, the longevity risk premium is far less than the interest rate and stock risk premiums. Later, we give optimal investment strategies in different scenarios with different values of $\theta_\lambda$. Pension contribution rates differ widely among schemes and countries. According to \cite{hmrc}, in the UK, there is a limit on the amount of tax-free pension savings that an individual can pay into his pension account in each tax year. However, there is no cap on the contribution rate. As stated in \cite{govuk}, the Department for Work and Pensions requires that the minimum contribution rate for DC schemes is $8\%$ under the UK legislation. We first consider $r_c=0.15$ then give a sensitivity analysis on the contribution rate. \cite{oecd} shows the net replacement rates vary across a large range from around $30\%$ to $90\%$ or more in OECD countries. The average net replacement rates of an average earner from mandatory schemes is $59\%$. In our base scenario, we adopt $r_w=0.59$. We set the instantaneous wage as $w=15$, thus the instantaneous contribution and annuity instalment are $c=2.25$ and $\pi=8.85$, respectively. We suppose the representative member's cumulated pension wealth at 40 years old is $F_0=50$ and the manager's risk aversion coefficient is $\gamma=2$. 

\begin{table}[htbp]
\begin{center}
\caption{Values of parameters in the base scenario}
{\begin{tabular}{lcccccc} \toprule
  Interest rate & Mortality & Stock  \\ \midrule
$r_0=0.0621328$ & $b=12.9374$ &	$\sigma_S=0.14926$ \\
$a_r=0.0056210$ & $\phi=0.0009944$ & $\sigma_S^r=-0.0046306$ \\
$b_r=0.0904668$	& $m=86.4515$ & $\theta_S=0.1108301$ \\
$\sigma_r=0.0543625$ & $b_\lambda=0.5610000$ \\
$\theta_r=-0.5590635$ & $\sigma_\lambda=0.0352$  \\ \bottomrule
\end{tabular}}
\label{tab:data}
\end{center}
\end{table}
We collect 100 simulation paths and show the average paths of optimal investment proportions in the left plot in Figure \ref{fig:base}. The right plot in Figure \ref{fig:base} displays the average path of $\frac{Y(t)}{F(t)}$. As discussed in Section \ref{sec:solution}, the ratio $\frac{Y(t)}{F(t)}$ and the optimal weights in stock and longevity bond decline over time. We see that the optimal weight in bond also decreases while the weight in cash rises. In the beginning, the short position in the money market reveals that the scheme manager borrows money and invests in risky assets to obtain risk premiums. It reveals that the manager takes an aggressive approach to quickly increase the pension scheme's wealth to a high level in the early stage. The reduction in the risky assets investment shows that, when closer to the retirement time, the manager becomes more conservative and shifts the scheme's wealth to safer assets. Moreover, the effect of interest rate and mortality fluctuations on the minimum guarantee reduces when approaching retirement time. This means the need for interest rate risk and longevity risk hedging becomes lower in the later investment period. Consequently, the manager cuts the investment proportions in bond and longevity bond. During the first five years, there is not much difference between the optimal weights in bond and longevity bond. However, the investment proportion in bond drops faster than the longevity bond. Besides, the portfolio is dominated by longevity bond throughout the investment horizon. Even at time $T$, the longevity bond weight is still very high at around $66.44\%$. This observation implies that the longevity bond provides an efficient way to hedge against both the interest rate and longevity risks. 
\begin{figure}[htbp]
    \centering
    \includegraphics[scale=0.45]{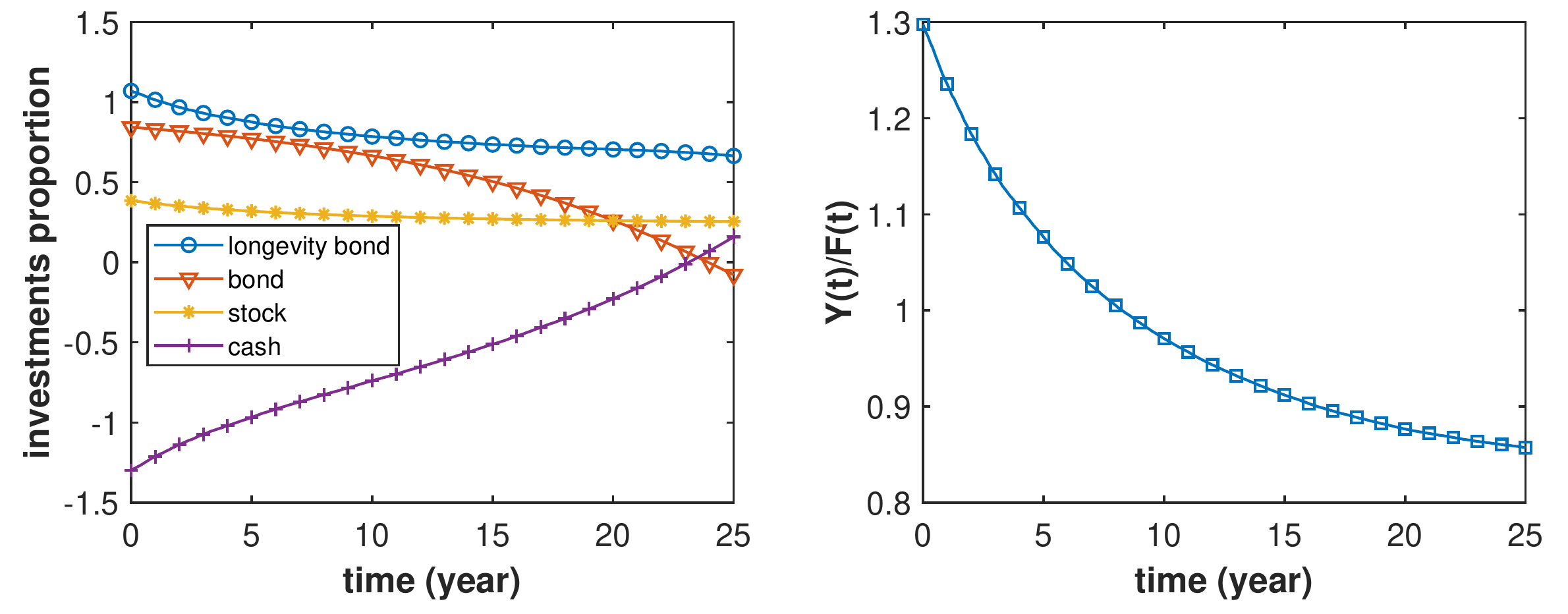}
    \caption{ Average paths of optimal investment proportions and $Y(t)/F(t)$ in the base scenario}
    \label{fig:base}
\end{figure}

\subsection{Sensitivity analysis}
\label{sec:sensitivity}
In Section \ref{sec:solution}, we give some comments on the impact of model parameters on the optimal investment strategy. This section provides various scenarios to investigate the impact of model parameters numerically. We are interested in the following parameters: risk aversion coefficient ($\gamma$), market price of longevity risk parameter ($\theta_\lambda$), maturity of rolling longevity bond ($T_L$), contribution rate ($r_c$) and wage replacement ratio ($r_w$). Other factors such as the market prices of interest rate risk and stock risk may also affect the optimal investment strategy sufficiently. Nonetheless, we do not look into those factors as the focus of our study is on hedging longevity risk.

\subsubsection{Risk-aversion}
In the power utility function \eqref{eq:utility}, $\gamma$ measures the relative risk-aversion. Higher the $\gamma$, more risk-averse the scheme manager. The four plots in Figure \ref{fig:gamma} show the average paths of optimal investment proportions with risk-aversion coefficient $\gamma$ equal to 2, 3, 4 and 5, respectively. Comparing the plots in Figure \ref{fig:base} and \ref{fig:gamma}, we observe that higher (or lower) the risk-aversion coefficient $\gamma$, higher (or lower) the investment proportions in bond and cash. Whereas, the investment proportions in longevity bond and stock decrease with $\gamma$. The bond is used to hedge the interest rate risk only  while the longevity bond provides interest rate risk as well as longevity risk hedge. A risk-averse investor tends to avoid relatively higher risk and prefers investments with lower risk but more guaranteed returns. Although the longevity bond provides higher risk premium, a scheme manager with a high risk-aversion coefficient prefers to invest more in safer assets such as bonds. The different strategies in these cases can also be explained by the optimal solution in Proposition \ref{Prop:solution}. The optimal investment proportions in bond and longevity bond are negatively correlated. In particular, we set $T_B=T_L=10$ and have $-\frac{f_1 (t,t+T_L )}{f_1 (t,t+T_B )}=-1$. This implies that the investment in longevity bond is fully taken from the investment in bond. Since the portfolio weight in longevity bond decreases with $\gamma$, the bond weight increases accordingly. This behaviour occurs as bond is a safer asset compared to longevity bond. The bottom right plot in Figure \ref{fig:gamma} reveals that, even for a highly risk-averse investor, the optimal investment proportion in longevity bond is almost always higher than $45\%$. Also, we find that the optimal weight in longevity bond is always higher than in stock, in all scenarios. Overall, we conclude that, even though highly risk-averse investors invest less in longevity bond, it is always an important component of the scheme's investment portfolio.
\begin{figure}[htbp]
    \centering
    \includegraphics[scale=0.45]{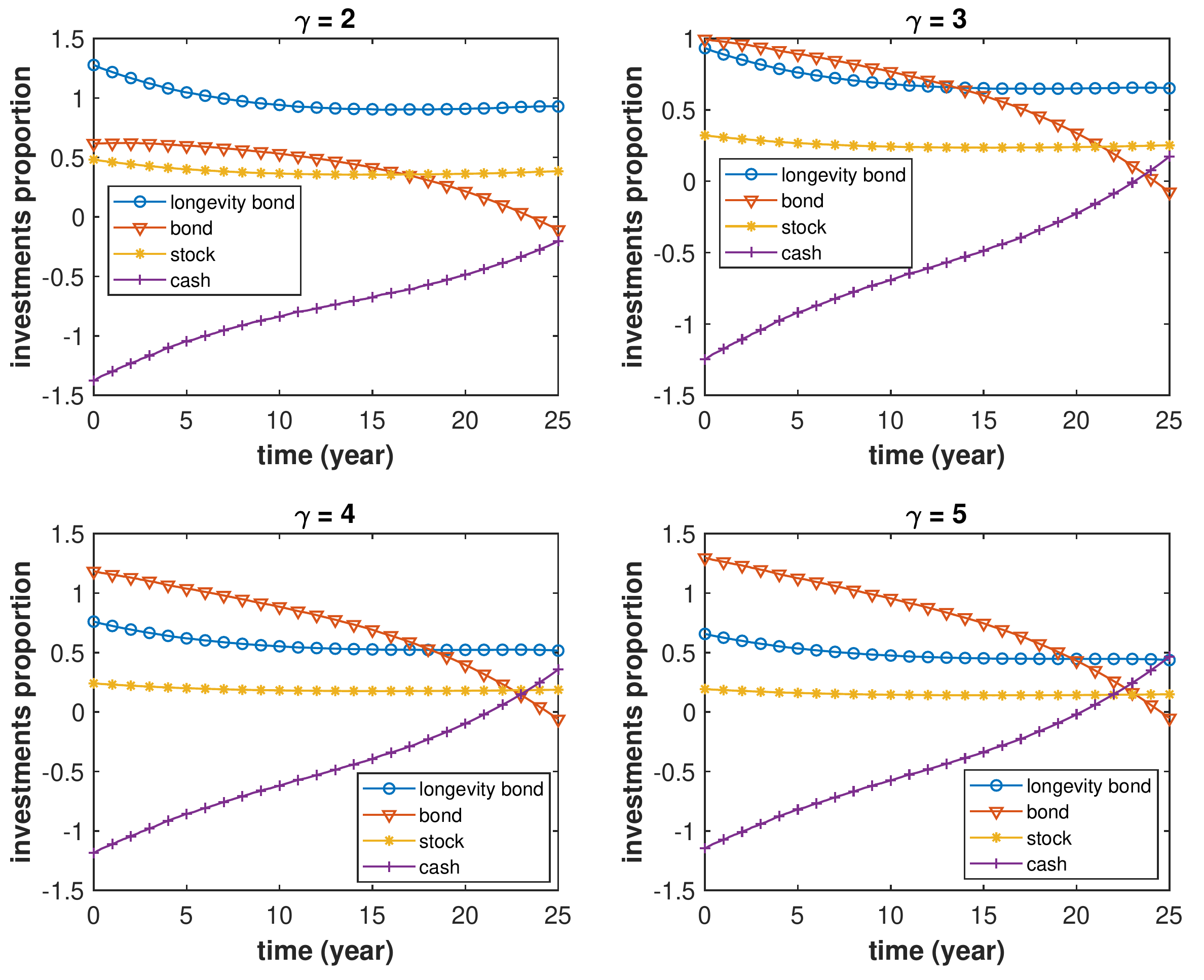}
    \caption{ Average paths of optimal investment proportions with $\gamma=2$, $3$, $4$ and $5$}
    \label{fig:gamma}
\end{figure}

\subsubsection{Market price of longevity risk}
The rolling bond provides an interest rate risk premium of $-f_1(t,t+T_B)\sigma_r\theta_rr(t)$. As the rolling longevity bond has the same maturity time as the rolling bond, it offers the same premium on the interest rate risk. Besides, the longevity bond provides a longevity risk premium of $-h_1(t,t+T_L)\sigma_\lambda\theta_\lambda \lambda(t)$. The stock offers an interest rate risk premium of $\sigma_S^r\theta_r r(t)$ and a stock risk premium of $\theta_S\sigma_S$. Clearly, $\theta_r$ and $\theta_S$ affect the optimal strategy. However, we are more interested in the impact of longevity risk premium and only provide an analysis on the market price of longevity risk parameter $\theta_\lambda$. From \eqref{eq:rollingL}, we learn that the longevity risk premium offered by the rolling longevity bond increases with $-\theta_\lambda$. Figure \ref{fig:theta} shows the average optimal strategies with $\theta_\lambda=-0.06$, $-0.08$, $-0.12$ and $-0.14$. At the initial time, the corresponding longevity risk premium in these cases are $1.1778\times10^{-5}$, $1.5703\times10^{-5}$, $2.3555\times10^{-5}$ and $2.7841\times10^{-5}$, respectively. In all these cases, the longevity risk premium is much smaller than the interest rate risk and stock risk premiums. The bottom right plot in Figure \ref{fig:theta} shows that when $\theta_\lambda=-0.14$, the optimal proportion in longevity bond is always around $1.0$ or above, and it always dominates the portfolio. The optimal proportion in stock is stable and a short position is taken in cash. After $t=21$, the scheme is financed by short-selling the bond. This is because when approaching the retirement time, the need to hedge against interest rate risk is insignificant compared to the need to hedge against the uncertain changes in mortality rate. The top left plot in Figure \ref{fig:theta} indicates that even in the case where the longevity risk premium is very low, the optimal weight in longevity bond is always higher than $45\%$. This reveals that the longevity bond is an important component in the pension scheme's investment portfolio, and is an efficient instrument to hedge the longevity risk. Comparing the plots, we observe that higher the longevity risk premium (that is, lower the $\theta_\lambda$) higher the portfolio weight in longevity bond. The interpretation is that with other parameters unchanged, a lower $\theta_\lambda$ increases the longevity risk premium but does not increase the uncertainty in longevity bond value. Thus, it makes the longevity bond more attractive for investment. The opposite reaction of bond weight and longevity bond weight against $\theta_\lambda$ is again explained by the discussion of Proposition \ref{Prop:solution}. The weight in cash is not sensitive to $\theta_\lambda$. This is because the changes in bond and longevity bond weights offset each other, and the total proportion in risky assets is not sensitive to $\theta_\lambda$. It is not a surprise that the optimal weight in stock is stable and is only slightly affected by $\theta_\lambda$. The reason can be inferred from the optimal solution that the portfolio weight on stock is $\frac{\theta_S}{\sigma_S \gamma }\frac{Y(t)}{F(t)}$. $\frac{\theta_S}{\sigma_S \gamma}$ does not depend on $\theta_\lambda$ and $\frac{Y(t)}{F(t)}$ is only slightly affected by $\theta_\lambda$. In summary, we conclude that the longevity bond hedges the scheme's longevity risk efficiently. Even when the longevity bond provides a relatively lower longevity risk premium, it is optimal to invest a large proportion of the scheme's wealth in it.
\begin{figure}[htbp]
    \centering
    \includegraphics[scale=0.45]{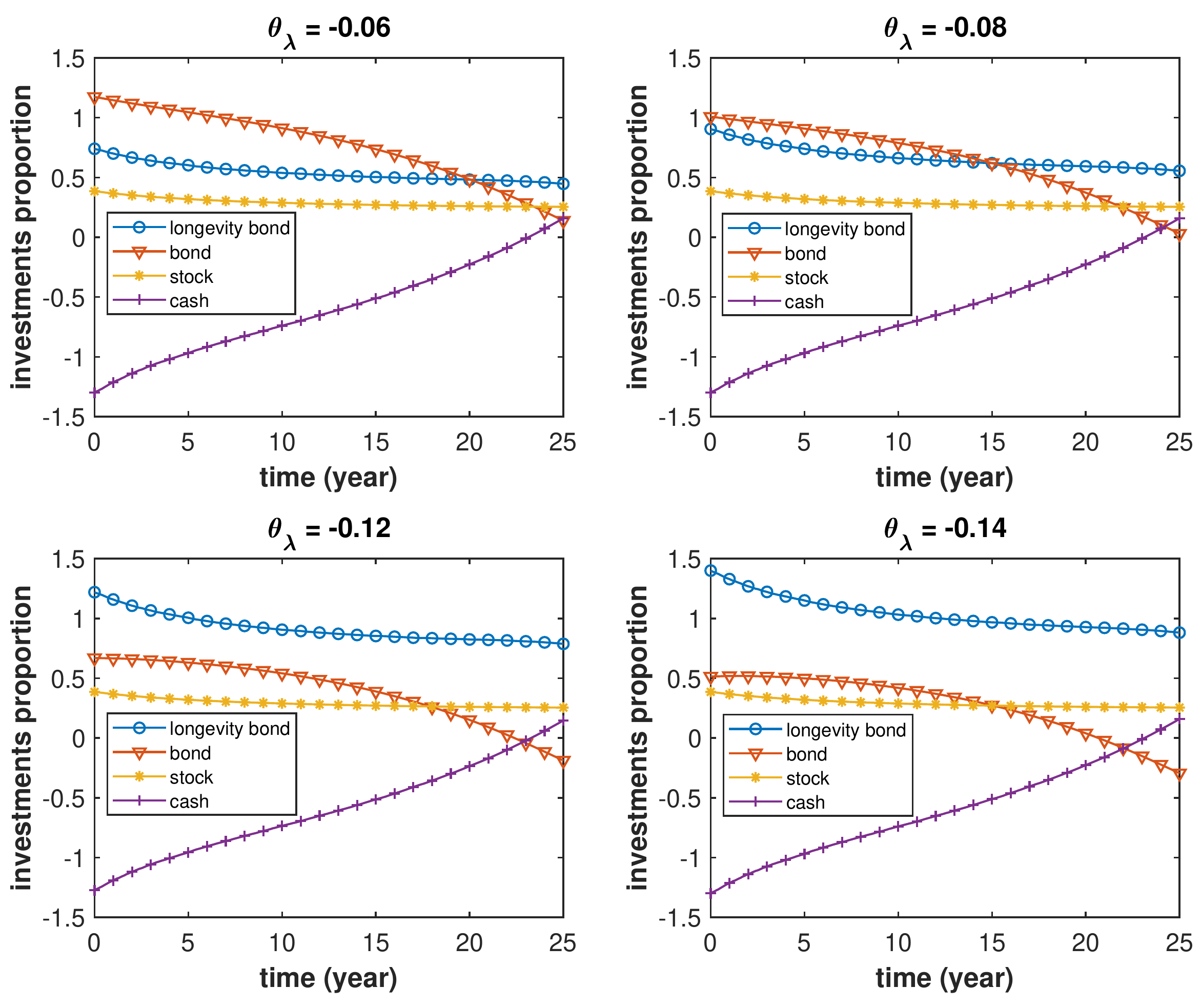}
    \caption{ Average paths of optimal investment proportions with $\theta_\lambda=-0.06$, $-0.08$, $-0.12$ and $-0.14$}
    \label{fig:theta}
\end{figure}

\subsubsection{Maturity of the rolling longevity bond}
\label{analysis:maturity}
In the previous numerical scenarios, we set $T_B=T_L=10$ and have $f_1 (t,t+T_B )=f_1 (t,t+T_L )$. In Figure \ref{fig:maturity}, we show the average optimal investment strategies with $T_L$ equal to 5, 15, 20 and 25. We observe that longer the rolling longevity bond maturity, lower the investment proportions in bond and longevity bond while more weight in cash. It seems that the optimal stock weight does not change much when $T_L$ increases. From Figure \ref{fig:base} and \ref{fig:maturity}, we see that although $w_L(t)$ decreases with $T_L$, the longevity bond always suppresses other assets. It is surprising that the optimal bond weight can reduce to around $-0.50$ when $T_L=25$ in late investment periods, meanwhile holding in the money market is positive. This indicates that the manager chooses to short sell bonds when approaching the retirement time. \cite{gao2008stochastic} and \cite{han2012optimal} also show similar findings. They showed that the pension portfolio would shift from investments in risky assets to the money market. The convention is that the bond guarantees a fixed amount of money at maturity. In the beginning, the weight put on bond is relatively high and is declining when closer to $T$ (the retirement time). It also implies that the need to hedge interest rate risk is lower when approaching $T$. We provide Figure \ref{fig:maturity2} to look into the impact of $T_L$ on the optimal weights in risky assets separately. From Proposition \ref{Prop:solution}, at any time $t$, the optimal stock weight depends on $\frac{Y(t)}{F(t)}$. The middle left plot in Figure \ref{fig:maturity2} shows that the optimal proportion in stock is not sensitive to $T_L$. Accordingly, there is a numerical evidence that $\frac{Y(t)}{F(t)}$ is not sensitive to $T_L$. 

We show in Figure \ref{fig:maturity3} that $\frac{1}{F(t)}$ does not change much when $T_L$ changes. We see that $\frac{1}{h_1(t,t+T_L)}$ is decreasing with $T_L$. According to the optimal solution, this results in the decline in longevity bond weight. At first, $\frac{1}{h_1(t,t+T_L)}$ drops dramatically when $T_L$ increases and later on changes only slightly. When $T_L$ equals to 5, 10, 15, 20 and 25, we have $\frac{1}{h_1(t,t+T_L)}$ equal to around 0.594886, 0.560673, 0.558716, 0.558597 and 0.558590, respectively. Therefore, we observe in Figure \ref{fig:maturity2} that the optimal weight in longevity bond is not sensitive to longer maturity times. $w_L$ shows an obvious decline when $T_L$ rises from 5 to 10. No distinguishable change in $w_L$ is observed when $T_L$ takes values 10, 15, 20 and 20. Intuitively, longer maturity time results in more uncertainty in the payment of rolling longevity bond. For a risk-averse investor, it is better to put less portfolio weight on longevity bond with a longer maturity. Compared to $w_L$, $w_B$ reacts more violently when changing $T_L$. This is apparent from the optimal solution as the changes in $w_L$ are scaled up by the factor $-\frac{f_1 (t,t+T_L ) }{f_1 (t,t+T_B)}$. The left plot in Figure \ref{fig:maturity3} shows $\frac{f_1 (t,t+T_L ) }{f_1 (t,t+T_B)}$ increases with $T_L$. Therefore, the optimal weight in bond decreases greatly with $T_L$ due to the deceasing $-\frac{f_1 (t,t+T_L ) }{f_1 (t,t+T_B)}w_L(t)$. The optimal investment proportion in bond becomes negative in the late years except for the case when $T_L=5$. This is because the longevity bond hedges longevity risk as well as interest rate risk. However, when closer to the retirement time, the need to hedge interest rate risk reduces while the longevity risk is still high. The negative positions in bond offset the interest rate risk hedge provided by the longevity bond. Even though negative positions in bond are observed, the sum of the weights in longevity bond and bond are always positive as shown in the bottom left plot in Figure \ref{fig:maturity2}. Overall, longer maturity times may lessen the attractiveness of longevity bond. Even though longer maturity times might be detrimental to its investment attractiveness, the longevity bond always plays an important role in DC scheme's risk management as its portfolio weight is always relatively high.

\begin{figure}[htbp]
    \centering
    \includegraphics[scale=0.45]{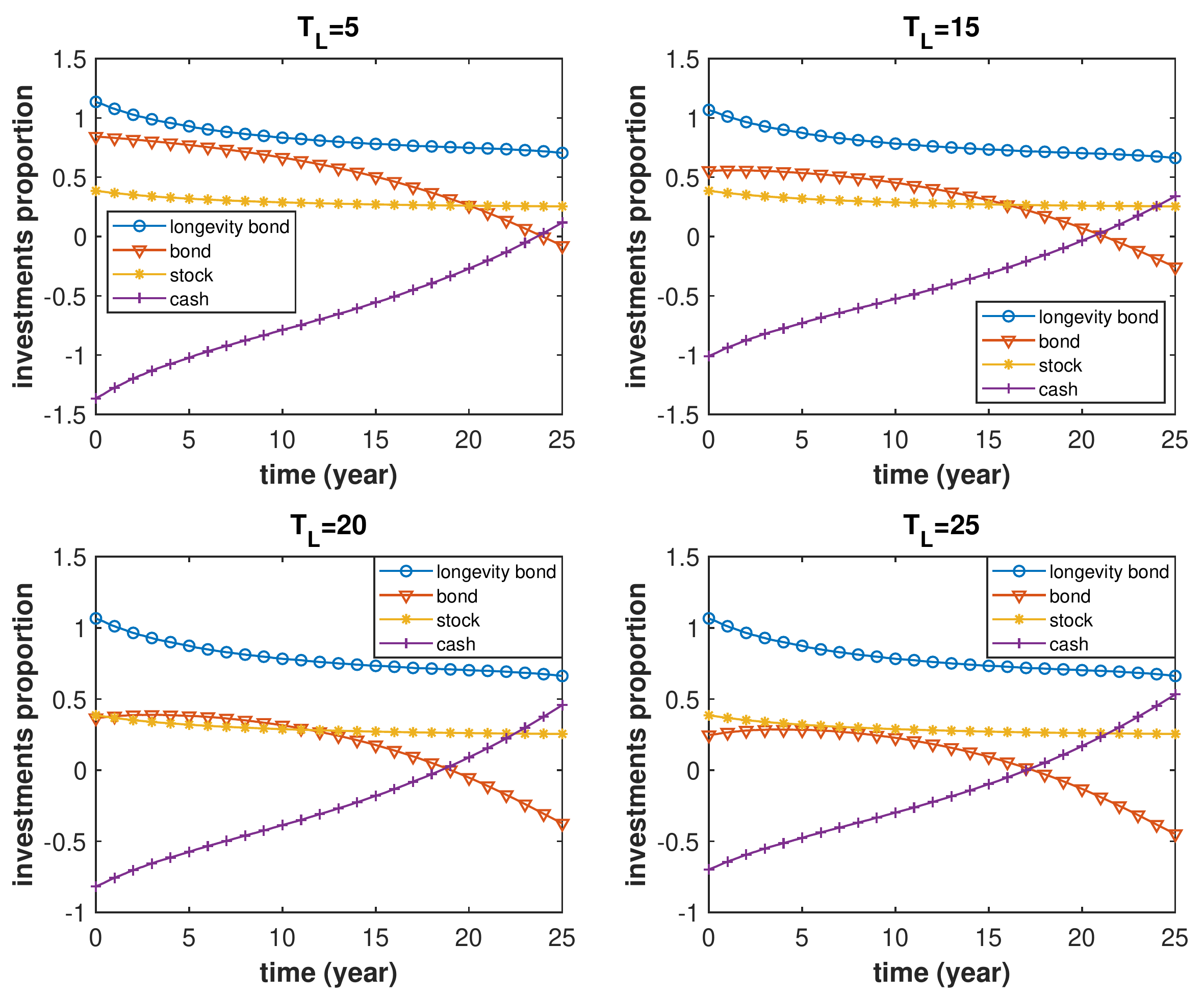}
    \caption{ Average paths of optimal investment proportions with $T_L=5$, $15$, $20$ and $25$ }
    \label{fig:maturity}
\end{figure}
\begin{figure}[htbp]
    \centering
    \includegraphics[scale=0.5]{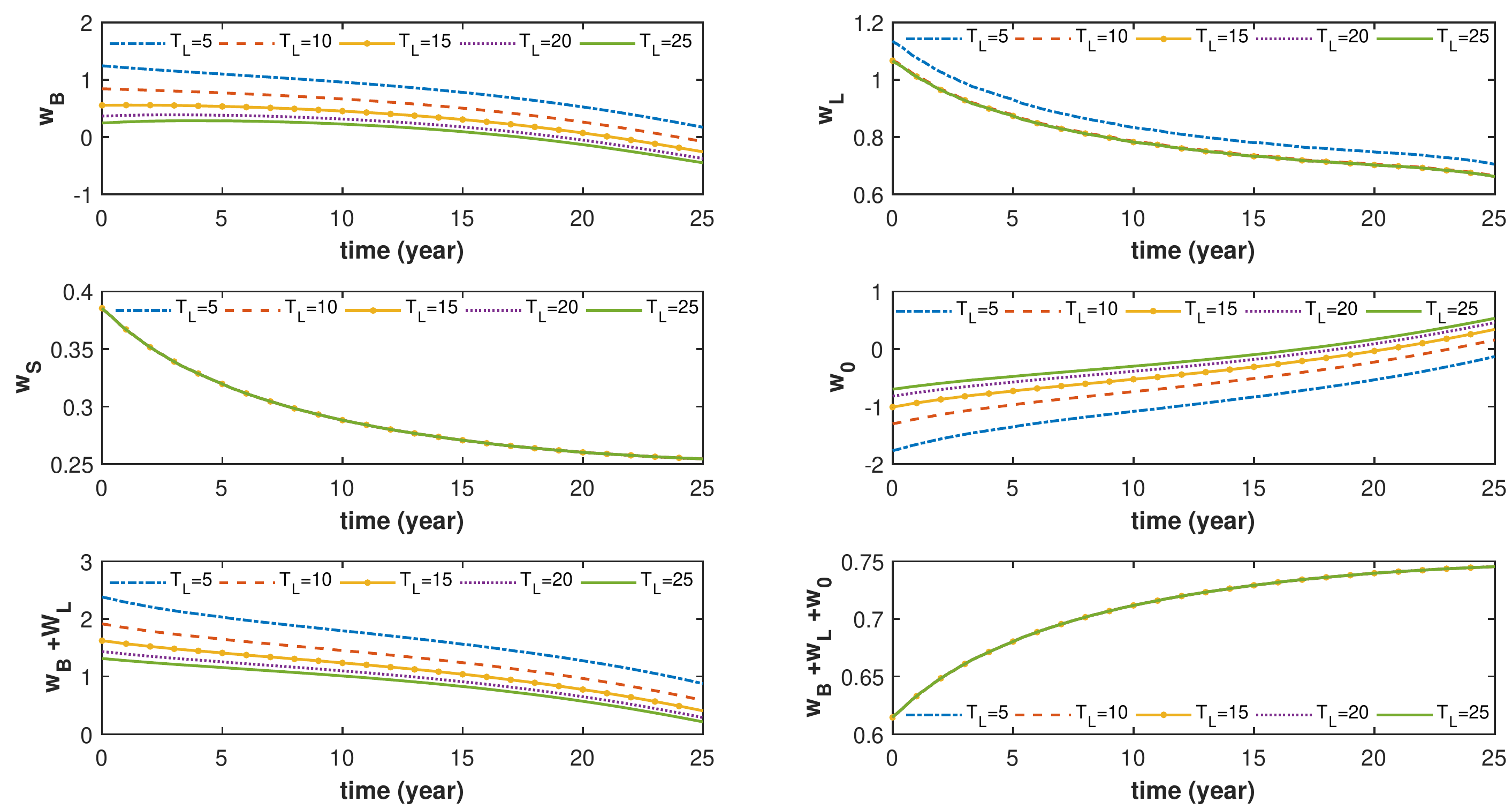}
    \caption{ Average paths of optimal investment proportions varying $T_L$ }
    \label{fig:maturity2}
\end{figure}
\begin{figure}[htbp]
    \centering
    \includegraphics[scale=0.45]{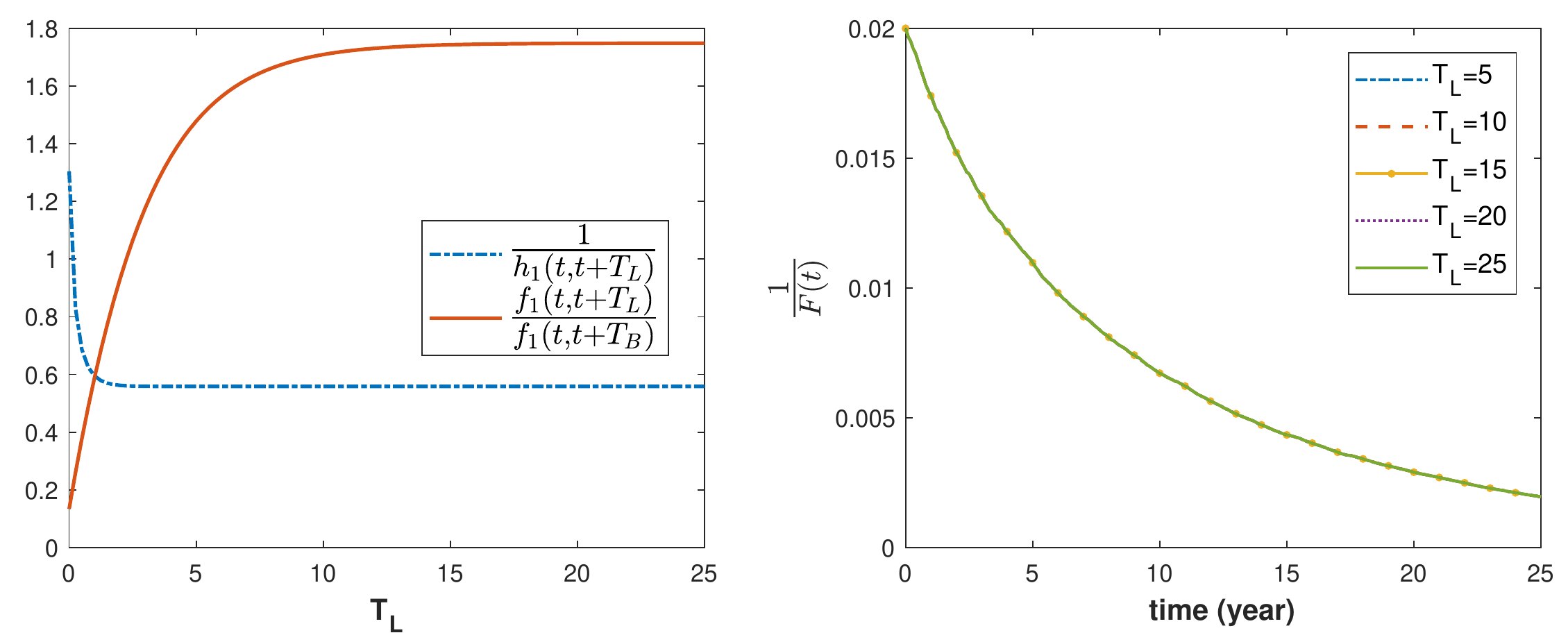}
    \caption{ $\frac{1}{h_1(t,t+T_L)}$ and $\frac{f_1(t,t+T_L)}{f_1(t,t+T_B)}$ versus $T_L$ ; Average path of $\frac{1}{F(t)}$ varying $T_L$ }
    \label{fig:maturity3}
\end{figure}

\subsubsection{Contribution rate}
\label{sec:contribution rate}
As stated in \cite{govuk} and \cite{oecd}, the minimum requirement on contribution rate for DC schemes is $8\%$, while there is no limit on the amount of contribution under UK legislation. For an eligible citizen, the government will add $1\%$ of the citizen's income into his pension in the form of tax relief resulting in an actual minimum contribution rate of $9\%$. An individual can contribute up to $100\%$ of his income into his pension scheme, but taxes may apply. Accordingly, we consider a maximum contribution rate of $40\%$. In Figure \ref{fig:contribution2}, we show the optimal proportions with contribution rate $r_c$ equals to 0.10, 0.20, 0.30 and 0.40. That is, the instantaneous contribution $c$ equals to 1.5, 3, 4.5 and 6. In general, we observe that higher the contribution rate, higher the weights in risky assets while lower the weight in cash. The intuition behind this observation is that higher the contribution rate, higher the present value of the scheme's future incomes. A high present value of the future incomes drives the manager to increase the investment in risky assets. Since there are guaranteed future incomes, the manager would like to take more risk and win more risk premiums. From Proposition \ref{Prop:solution}, the optimal stock weight depends on $\frac{Y(t)}{F(t)}$. At any time $t$, we have $Y(t)=F(t)+D(t)-G(t)$. More contribution paid into the scheme results in higher wealth level $F(t)$ and greater present value of future contributions $D(t)$. Besides, the present value of minimum guarantee $G(t)$ does not depend on the contribution rate. Consequently, $\frac{Y(t)}{F(t)}$ and the optimal investment proportion in stock increase with $r_c$. When closer to the retirement time, the optimal stock weight is less sensitive to $r_c$. The interpretation is that the wealth process $F(t)$ increases over time as there is continuous contribution paid into the scheme. Besides, the scheme receives investment returns. The expected value of future contributions $D(t)$ decreases over time and eventually reaches zero at $T$. Therefore, $\frac{Y(t)}{F(t)}$ drops over time and the impact of $r_c$ on optimal stock weight reduces. From the optimal solution, the optimal investment proportion in longevity bond at any time $t$ is 
\eqstar{
w_L^\star(t)=-\frac{\theta_\lambda+\sigma_\lambda A_2(t,T)}{\gamma\sigma_\lambda h_1 (t,t+T_L)}\frac{Y(t)}{F(t)}+\frac{\pi  \int_T^\infty L(t,s)h_1(t,s)\dd s}{h_1(t,t+T_L)}\frac{1}{F(t)}
}
Since $\frac{Y(t)}{F(t)}$ increases with $r_c$ and $\frac{1}{F(t)}$ decreases with $r_c$, it is hard to see from the optimal solution how $w_L(t)$ changes with $r_c$. From the numerical results, we observe that a higher $r_c$ leads to more weight in longevity bond in the first 17 years and $w_L(t)$ becomes less sensitive to $r_c$ in later period. We infer that in the early years, the first term in the above formula increases faster and the second term decreases slower when $r_c$ rises. Thus, $w_L(t)$ increases with $r_c$. When approaching $T$, the movements in the two terms offset each other gradually. Consequently, insensitivity is observed. The optimal weight in bond reacts in a similar manner with longevity bond, though far less sensitively. We also find that, when approaching $T$, both the proportions of the wealth invested in risky assets and cash are less sensitive to the contribution rate. Generally, our numerical results imply that it is optimal to put more weights on risky assets, especially longevity bond, if the contribution rate is high.

\begin{figure}[htbp]
    \centering
    \includegraphics[scale=0.45]{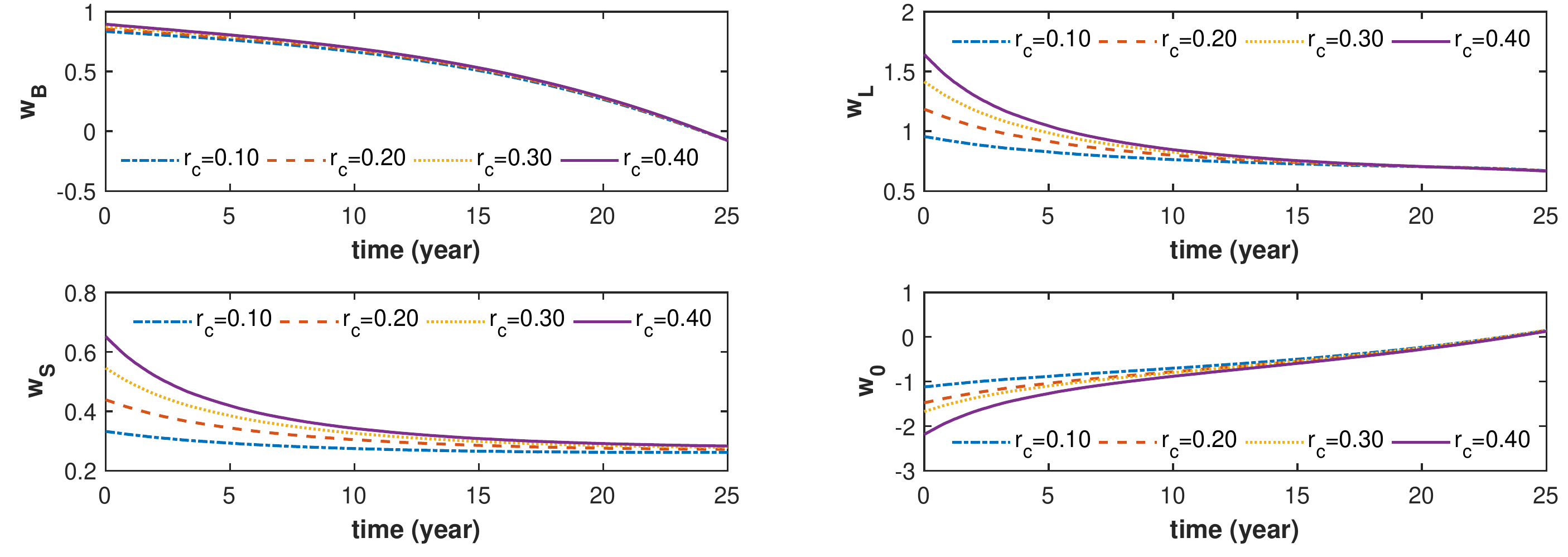}
    \caption{ Average paths of optimal investment proportions with $r_c=0.10$, $0.20$, $0.30$ and $0.40$}
    \label{fig:contribution2}
\end{figure}

\subsubsection{Wage replacement ratio}
\label{sec:wage replacement}
The wage replacement ratio is a good tool when estimating retirement income needs. A high wage replacement ratio implies that a high fraction of the pre-retirement income is needed to maintain living standard in retirement. \cite{oecd} reveals that the net replacement ratio varies from $30\%$ to $90\%$ among OECD countries. Accordingly, we set $r_w$ equal to 0.30, 0.50, 0.70 and 0.90 to test the impact of wage replacement ratio. In Figure \ref{fig:replacement2}, we look into the effect of $r_w$ on the optimal weights in bond, longevity bond, stock and cash separately. We find that the optimal weights in bond and cash are less responsible to changes in the wage replacement ratio. The optimal stock weight declines with the wage replacement ratio. While a high wage replacement ratio results in more fraction of the scheme wealth invested in longevity bond. It is straightforward to see that higher the wage replacement ratio, higher the annuity instalments and higher the $G(t)$. While the wealth process $F(t)$ and the discounted future contributions $D(t)$ do not depend on the wage replacement ratio. Thus, $\frac{Y(t)}{F(t)}=\frac{F(t)+D(t)-G(t)}{F(t)}$ declines with $r_w$. From Proposition \ref{Prop:solution}, we see that the optimal stock weight increases with $\frac{Y(t)}{F(t)}$ while the optimal longevity bond weight decreases with $\frac{Y(t)}{F(t)}$. The results reveal that it is optimal to increase the fraction of scheme wealth in longevity bond when the wage replacement ratio is high. Intuitively, a high wage replacement ratio implies the members require high annuity instalments and the scheme is exposed to greater longevity risk. As a consequence, the scheme manager invests a large proportion of the scheme's wealth in longevity bond to hedge the longevity risk.
\begin{figure}[htbp]
    \centering
    \includegraphics[scale=0.53]{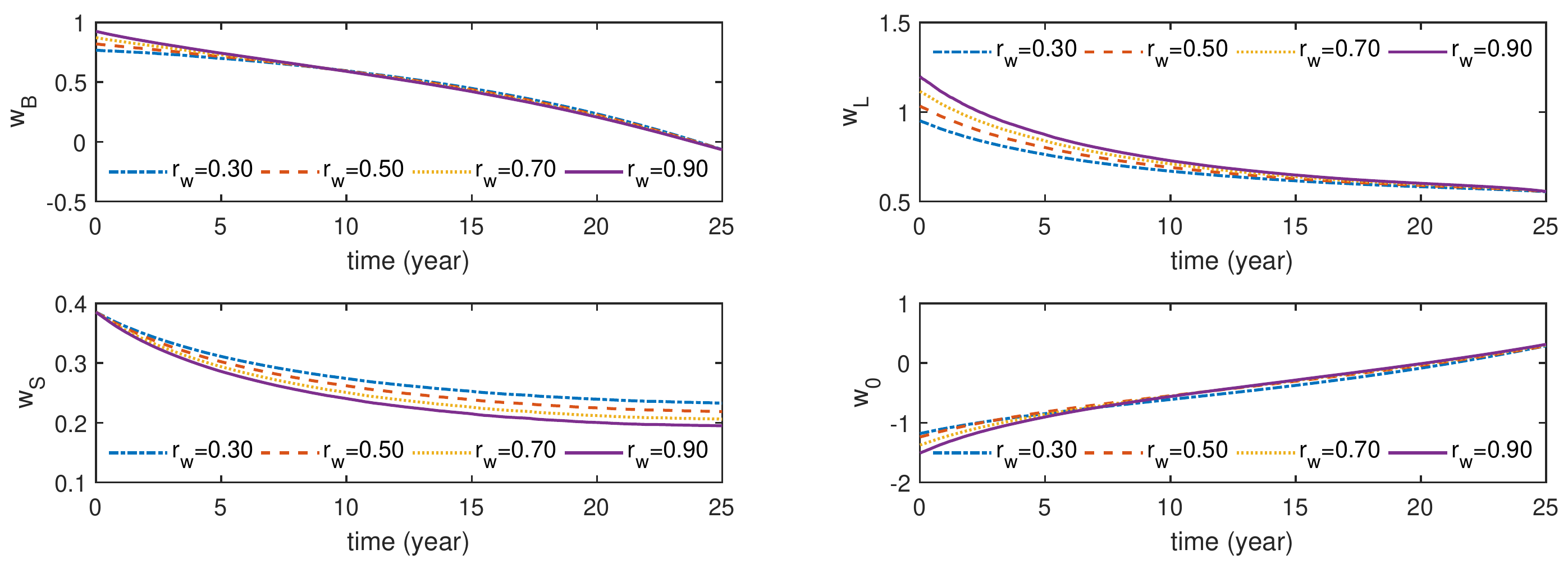}
    \caption{Average paths of optimal investment proportions with $r_w=0.30$, $0.50$, $0.70$ and $0.90$}
    \label{fig:replacement2}
\end{figure}

\section{Conclusion}
\label{sec:conclusion}
We studied the optimal investment problem for a DC pension scheme in a framework where both interest rate risk and longevity risk are considered. Our theoretical results and subsequent numerical studies showed evidence that the longevity bond plays an important role in DC scheme's risk management. We observed that more risk-averse the scheme manger, lower the investment proportion in longevity bond. However, even for a highly risk-averse manager, we showed that it is optimal to invest a large proportion of the scheme's wealth in the longevity bond. Also, compared to the investment in other risky assets, the investment proportion in longevity bond is shown to be relatively high even when the longevity risk premium is relatively low. Moreover, we observed that longer maturity times could be detrimental to the attractiveness of the longevity bond, however longevity bonds with longer maturity times always dominate the investment portfolio. Besides, we observed that high contribution rates and wage replacement ratios urge the scheme manager to invest more in the longevity bond. We conclude that the longevity bond is attractive to pension schemes and there is genuine potential in the development of mortality-linked derivatives.

\begin{appendix}
\section{Detailed calculations in Section \ref{sec:framework} }
\label{calculations}
\subsection{Pricing zero-coupon bond}
\label{sec:app1}
For any $t\in[0,T]$, under the risk-neutral probability measure $\widetilde{\Pb}$, the dynamics of the risk-free interest rate is given by
\eqstar{
\dd r(t) = \left( a_r-\tilde{b}_r r(t) \right)\dd t+\sigma_r \sqrt{r(t)} \dd \widetilde{W}_1 (t),
}
where $\tilde{b}_r=b_r+\theta_r\sigma_r$. Let $k(t,r(t))=B(t,T_B)$ and define the interest rate discount process $C(t)=e^{-\int_0^t r(u) \dd u}$. The discounted price of the zero-coupon bond $C(t)B(t,T_B)=C(t)k(t,r(t))$ is a martingale under the risk-neutral measure $\widetilde{\Pb}$. From It\^{o}'s formula, we get
\eqstar{
\dd \big(C(t)k(t,r(t))\big) &= k(t,r(t))\dd C(t)+C(t)\dd k(t,r(t))\\
& =  -rkC\dd t+C\Big[k_t \dd t+k_r \dd r+\frac{1}{2}k_{rr}\dd r\dd r \Big]\\
& = C\Big[-rk+k_t+k_r(a_r-\tilde{b}_r r)+ \frac{1}{2}k_{rr}\sigma_r^2 r \Big]\dd t+k_r\sigma_r\sqrt{r}C\dd \Wt_r.
}
Setting $\dd t$ term equal to zero leads to
\eqlnostar{PDE:gbond}{
k_t+k_r(a_r-\tilde{b}_r r)+ \frac{1}{2}k_{rr}\sigma_r^2 r=rk,
}
with the terminal condition $k(T_B,r) = 1$. As $r(t)$ follows an affine model, the solution of \eqref{PDE:gbond} could be expressed in the following form
\eqlnostar{eq:affinebond}{
k(t,r(t)) = e^{f_0(t,T_B)-f_1(t,T_B)r(t)},
}
with terminal conditions $f_0(T_B,T_B)=0$ and $f_1(T_B,T_B)=0$. From \eqref{PDE:gbond} and \eqref{eq:affinebond}, we derive
\eqstar{
f_0^\prime(t,T_B)-f_1^\prime (t,T_B) r(t)-f_1(t,T_B)(a_r-\tilde{b}_r r(t))+\frac{1}{2}f_1^2(t,T_B)\sigma_r^2 r(t)=r(t).
}
By collecting $r(t)$ terms above, we get the following system of ODEs
\eqstar{
&1=-f_1^\prime (t,T_B)+f_1(t,T_B)\tilde{b}_r +\frac{1}{2}f_1^2(t,T_B)\sigma_r^2 ,\\
&0=f_0^\prime(t,T_B)-f_1(t,T_B)a_r.
}
Solving the ODEs (for example, by using MATLAB), we obtain the solutions $f_0(t,T_B)$ and $f_1(t,T_B)$ in \eqref{eq:bondf}. Then, we obtain
\eqlnostar{eq:dBQ}{
\frac{\dd B(t,T_B)}{B(t,T_B)} =& r(t)\dd t-f_1(t,T_B)\sigma_r\sqrt{r(t)}\dd \widetilde{W}_1(t)\\
=&r(t)\dd t-f_1(t,T_B)\sigma_r\sqrt{r(t)}\left(\dd W_1(t)+\theta_r\sqrt{r(t)}\dd t\right).
}

\subsection{Pricing zero-coupon longevity bond}
\label{sec:appendix2}

For any $t\in[0,T]$, under the risk-neutral measure $\widetilde{\Pb}$, the dynamics of the force of mortality is given by
\eqstar{
    \dd \lambda(t)=\left(a_{\lambda} (t)-\tilde{b}_{\lambda} \lambda(t)\right)\dd t +\sigma_{\lambda}\sqrt{\lambda(t)}\dd \widetilde{W}_2 (t),
}
where $\tilde{b}_\lambda=b_\lambda +\theta_\lambda\sigma_\lambda $. Denote by $\tilde{C}(t) = e^{-\int_0^t \lambda(u)\dd u}$ and $\Bt(t,T_L)=\widetilde{\Eb} \left[e^{-\int_t ^{T_L} \lambda(u)\dd u} \biggm| \Fc(t)\right]$. Let $\tilde{k}(t,\lambda(t)) = \Bt(t,T_L)$, we have that $\Ct(t)\Bt(t,T_L)=\Ct(t)\tilde{k}(t,\lambda(t))$ is a martingale under $\Pbt$. We obtain by It\'o's formula
\eqlnostar{PDE:glongevity}{
\tilde{k}_t+\tilde{k}_\lambda(a_\lambda-\tilde{b}_\lambda \lambda)+ \frac{1}{2}\tilde{k}_{\lambda\lambda}\sigma_\lambda^2 \lambda =\lambda\tilde{k},
}
with the terminal condition $\tilde{k}(T_L,\lambda) = 1$. Due to the affine nature of $\lambda(t)$, the solution of \eqref{PDE:glongevity} can be expressed in the following form
\eqlnostar{eq:affinelongevity}{
\tilde{k}(t,\lambda(t)) = e^{h_0(t,T_L)-h_1(t,T_L)\lambda(t)},
}
with terminal conditions $h_0(T_L,T_L)=0$ and $h_1(T_L,T_L)=0$. From \eqref{PDE:glongevity} and \eqref{eq:affinelongevity}, we obtain
\eqstar{
 &h_0^\prime(t,T_L)-h_1^\prime (t,T_L) \lambda(t)-h_1(t,T_L)(a_\lambda(t)-\tilde{b}_\lambda \lambda(t))+\frac{1}{2}h_1^2(t,T_L)\sigma_\lambda^2 \lambda(t)=\lambda(t) .
}
By collecting the $\lambda(t)$ terms above, we obtain the following system of ODEs
\eqstar{
&1 = -h_1^\prime (t,T_L)+h_1(t,T_L)\tilde{b}_\lambda +\frac{1}{2}h_1^2(t,T_L)\sigma_\lambda^2 ,\\
&0 = h_0^\prime(t,T_L)-h_1(t,T_L)a_\lambda(t).
}
The solutions are given in \eqref{eq:bondh}. From It\^{o}'s formula, we get
\eqstar{
\frac{\dd \Bt(t,T_L)}{\Bt(t,T_L)} = &\lambda(t)\dd t -h_1(t,T_L)\sigma_\lambda\sqrt{\lambda(t)}\dd \widetilde{W}_2(t).
}
Since $r(t)$ and $\lambda(t)$ are independent, we can rewrite $L(t,T_L)$ as
\eqstar{
    L(t,T_L)=e^{-\int_0^t \lambda(u)\dd u}\widetilde{\Eb}\left[e^{-\int_t ^{T_L} r(u)\dd u}\biggm| \Fc(t)\right]\widetilde{\Eb}\left[e^{-\int_t ^{T_L} \lambda(u)\dd u} \biggm| \Fc(t)\right].
}
Denote by $N(t,T_L) = B(t,T_L)\Bt(t,T_L)$, we have $L(t,T_L)=\tilde{C}(t)N(t,T_L)$ which gives us the following
\eqstar{
\dd L(t,T_L) =& N(t,T_L)\dd \tilde{C}(t)+\tilde{C}(t)\dd N(t,T_L)\\
= & N(t,T_L)\dd \tilde{C}(t)+\tilde{C}(t)\Big[B(t,T_L)\dd \Bt(t,T_L)+\Bt(t,T_L)\dd B(t,T_L)\Big]\\
=& r(t)L(t,T_L)\dd t-f_1(t,T_L)\sigma_r\sqrt{r(t)}L(t,T_L) \left(\dd W_1(t)+\theta_r\sqrt{r(t)}\dd t\right)\\
&-h_1(t,T_L)\sigma_\lambda\sqrt{\lambda(t)}L(t,T_L)\left(\dd W_2(t)+\theta_\lambda\sqrt{\lambda(t)}\dd t\right).
}
\subsection{Dynamics of the scheme's wealth process}
\label{appendix:wealth}
To study the dynamics of the scheme's wealth, we employ a similar method used in \cite{he2013optimalb}. For any $t\in[0,T]$ and a small positive number $\Delta$, denote by $\nu(t,t+\Delta) $ the rate of investment return in the time interval $(t,t+\Delta)$, we have
\eqstar{
\nu(t,t+\Delta)F(t)=&\alpha_0(t)\frac{R(t+\Delta)-R(t)}{R(t)}+\alpha_B(t)\frac{B(t+\Delta)-B(t)}{B(t)}+\alpha_L(t) \frac{L(t+\Delta)-L(t)}{L(t)}\\
&+\alpha_S(t)\frac{S(t+\Delta)-S(t)}{S(t)}.
}
Let $q(t,t+\Delta)$ denote the fraction of the members who pass away in the time interval $(t,t+\Delta)$. The total amount of passed away members' pension taken out by the heirs equals to $q(t,t+\Delta)F(t)$. The scheme's cashflow in this period is $c\Delta+\nu(t,t+\Delta)F(t)-q(t,t+\Delta)F(t)$. At time $t+\Delta$, $1-q(t,t+\Delta)$ fraction of the members survive and the fund value is
\eqstar{
F(t+\Delta)=\Big(F(t)+c\Delta+\nu(t,t+\Delta)F(t)-q(t,t+\Delta)F(t)\Big)\frac{1}{1-q(t,t+\Delta)}.
}
The Taylor series approximation of $1-q(t,t+\Delta)=e^{-\int_t^{t+\Delta}\lambda(u)\dd u}$ gives
\eqstar{
\frac{1}{1-q(t,t+\Delta)}=e^{\int_t^{t+\Delta}\lambda(u)\dd u}=1+\lambda(t)\Delta +o(\Delta).
}
Since $\nu (t,t+\Delta)\lambda(t)\Delta=o(\Delta)$ and $q(t,t+\Delta)\lambda(t)\Delta=o(\Delta)$, we get
\eqstar{
F(t+\Delta)=&\Big(F(t)+c\Delta+\nu(t,t+\Delta)F(t)-q(t,t+\Delta)F(t)\Big)\left(1+\lambda(t)+o(\Delta)\right)\\
=&F(t)+c\Delta+\nu(t,t+\Delta)F(t)-q(t,t+\Delta)F(t)+\lambda(t)\Delta F(t)+o(\Delta).
}
Since we have,
\eqstar{
\underset{\Delta\to 0}{\lim} \frac{q(t,t+\Delta)}{\Delta}=&\lambda(t),\\
\underset{\Delta\to 0}{\lim} \nu(t,t+\Delta)F(t)=&\alpha_0(t)\frac{\dd R(t)}{R(t)}+\alpha_B(t)\frac{\dd B(t)}{B(t)}+\alpha_L(t)\frac{\dd L(t)}{L(t)}+\alpha_S(t)\frac{\dd S(t)}{S(t)},
}
we obtain
\eqstar{
\dd F(t)=&\left(\lambda(t)F(t)+c -\lambda(t)F(t)\right)\dd t+\alpha_0(t)\frac{\dd R(t)}{R(t)}+\alpha_B(t)\frac{\dd B(t)}{B(t)}+\alpha_L(t)\frac{\dd L(t)}{L(t)}+\alpha_S(t)\frac{\dd S(t)}{S(t)}\\
=&\Big(r(t)F(t)+c+\alpha(t)^\prime M(t)\Big)\dd t+\alpha(t)^\prime \Sigma(t)^\prime \dd W(t).
}

\section{Proof of Proposition \ref{prop:D}}
\label{appendix:dD}
\begin{proof}
At ant time $t\in[0,T]$, by interchanging the order of integration, we can rewrite \eqref{eq:future contributions} as
\eqstar{
D(t) =\Ebt\left[\int_t^T ce^{-\int_t ^s r(u)\dd u} \dd s\Biggm| \Fc(t)\right]=c\int_t^T B(t,s)\dd s.
}
From Leibniz integral rule, we get
\eqstar{
\frac{\dd D(t)}{\dd t} = -c+c\int_t^T \frac{\dd B(t,s)}{\dd t}\dd s.
}
Then, we obtain
\eqlnostar{eq:dD}{
\dd D(t) = &-c\dd t+c\int_t^T B(t,s)\frac{\dd B(t,s)}{B(t,s)}\dd s\nonumber \\
= & -c\dd t+r(t)D(t)\dd t+c\int_t^T B(t,s) \sigma_B(t,s)\dd s\left(\dd W_1(t)+\theta_r\sqrt{r(t)} \dd t\right).
}
Comparing the coefficients in \eqref{eq:rollingbond} and \eqref{eq:dD}, we obtain the holdings in rolling bond and money market:
\eqstar{
\alpha_B ^D (t)&=\frac{c\int_t ^TB(t,s)\sigma_B(t,s)\dd s  }{\sigma_B(t,t+T_B)}=\frac{ c\int_t ^TB(t,s)f_1 (t,s)\dd s}{f_1 (t,t+T_B )},\\
\alpha_B ^0 (t)&=D(t)-\alpha_B ^D (t).
}
\end{proof}
\section{Proof of Proposition \ref{prop:G}}
\label{appendix:dG}
\begin{proof}
At any time $t\in[0,T]$, we can rewrite \eqref{eq:definition guarantee} as
\eqstar{
G(t) =&\pi \Ebt \left[\Ebt \left[\int_T ^\infty  \frac{R(t)}{R(s)} p(s)\dd s \biggm| \Fc(T)\right]\biggm|\Fc(t)\right]\\
=&\pi  \int_T ^\infty \Ebt \left[\frac{R(t)}{R(s)} p(s)\biggm|\Fc(t)\right]\dd s= \pi  \int_T^\infty L(t,s)\dd s
}
From Leibniz integral rule, we get
\eqstar{
\frac{\dd G(t)}{\dd t} = \pi  \int_T^\infty\frac{\dd L(t,s)}{\dd t}\dd s.
} 
Then, we obtain
\eqlnostar{eq:dG}{
\dd G(t)=& \pi \int_T^\infty \frac{\dd L(t,s)}{L(t,s)}L(t,s)\dd s\\
=& \pi \int_T^\infty \left[r(t)\dd t+\sigma_L^r (t,s )\dd \Wt_1 (t)+\sigma_L ^\lambda(t,s)\dd \Wt_2 (t)\right]L(t,s)\dd s\\
= &r(t)G(t)\dd t+\pi \int_T^\infty\sigma_L^r(t,s)L(t,s)\dd s\left(\dd W_1(t) +\theta_r\sqrt{r(t)}\dd t\right)\\
&+\pi  \int_T^\infty\sigma_L^\lambda(t,s)L(t,s)\dd s \left(\dd W_2(t)+\theta_\lambda\sqrt{\lambda(t)}\dd t\right).
}
By comparing the coefficients in \eqref{eq:rollingbond}, \eqref{eq:rollingL} and \eqref{eq:dG}, we get the holdings in rolling longevity bond, rolling bond and money market:
\eqstar{
\alpha_L^G (t)=&\frac{\pi \int_T^\infty\sigma_L^\lambda(t,s)L(t,s)\dd s}{\sigma_L^\lambda(t,t+T_L)}=\frac{\pi \int_T^\infty L(t,s)h_1(t,s)\dd s}{h_1(t,t+T_L)},\\
\alpha_B^G(t) =& \frac{\pi \int_T^\infty\sigma_L^r(t,s)L(t,s)\dd s}{\sigma_B(t,t+T_B)}-\alpha_L^G (t)\frac{\sigma_L^r(t,t+T_L)}{\sigma_B(t,t+T_B)}\\
=&\frac{\pi \int_T^\infty L(t,s)f_1(t,s)\dd s}{f_1(t,t+T_B)}-\alpha_L^G (t)\frac{f_1(t,t+T_L)}{f_1(t,t+T_B)},\\
\alpha_0^G(t) = &G(t)-\alpha_B^G(t)-\alpha_L^G(t).
} 
\end{proof}
\section{Proof of Proposition \ref{Prop:solution1}}
\label{sec:appendixc}
\begin{proof}
For any $t\in[0,T]$, let $g(t,z)$ be a function of $t\ \text{and}\ z(t).$ We guess the solution of the second order non-linear partial differential equation (\ref{eq:hjb}) is of the following form
\eqlnostar{eq:vg}{
V(t,y,z)=\frac{y^{1-\gamma}}{1-\gamma}g(t,z),
}
with terminal condition $g(T,z) = 1$. Substituting \eqref{eq:vg} in \eqref{eq:hjb} leads to
\eqlnostar{eq:partialg}{
    0=&g_t+(1-\gamma)rg+\frac{1-\gamma}{2\gamma} M^\prime (\Sigma^\prime \Sigma)^{-1} Mg
    +\frac{1-\gamma}{\gamma} M^\prime  \Sigma^{-1} \xi g_z+\mu^\prime g_z+\frac{1}{2}tr(\xi^\prime \xi g_{zz} )+\frac{1-\gamma}{2\gamma g} {g_z}^\prime \xi^\prime \xi g_z.
}
We further guess that $g(t,z)$ is of the following form
\eqlnostar{eq:g}{
    g(t,z)=e^{A_0(t,T)+A(t,T)z(t)}=e^{A_0 (t,T)+A_1 (t,T)r(t)+A_2(t,T)\lambda(t)}
}
with terminal conditions $A_0 (T,T)=0$, $A_1 (T,T)=0$ and $A_2(T,T)=0$. Substituting \eqref{eq:g} in \eqref{eq:partialg}, we have
\eqstar{
0= &\left(A_0^\prime+A_1^\prime r+A_2^\prime \lambda\right)+(1-\gamma)r+\frac{1-\gamma}{2\gamma}\left(\theta_r^2 r+\theta_\lambda^2\lambda+\theta_S^2 \right)+\frac{1-\gamma}{\gamma}\left(\theta_r\sigma_r rA_1+\theta_\lambda\sigma_\lambda \lambda A_2\right)\\
&+\left( a_r-b_r r\right)A_1+\left( a_\lambda-b_\lambda \lambda\right)A_2+\frac{1}{2\gamma}\left( \sigma_r^2 r A_1^2+\sigma_\lambda^2 \lambda A_2^2\right).
}
By collecting the $r(t)$ and $\lambda(t)$ terms above, we obtain the following three ODEs:
\eqlnostar{eq:ODEs}
{
&0=A_1^\prime(s,T)+\frac{(1-\gamma)(2\gamma+\theta_r^2)}{2\gamma} +\frac{(1-\gamma)\theta_r \sigma_r-b_r\gamma}{\gamma} A_1 (s,T)+\frac{\sigma_r^2}{2\gamma}  A_1^2(s,T),\\
& 0= A_2^\prime(s,T)+\frac{(1-\gamma)\theta_\lambda^2}{2\gamma} +\frac{(1-\gamma) \theta_\lambda \sigma_\lambda-b_\lambda \gamma}{\gamma} A_2(s,T)+\frac{\sigma_\lambda^2}{2\gamma}  A_2(s,T)^2,\\
&0=A_0^\prime(s,T)+\frac{1-\gamma}{2\gamma} \theta_S^2 +a_r A_1(s,T) +a_\lambda(s) A_2.
} 
Under conditions $\Delta_1>0$ and $\Delta_2>0$, the solutions $A_0(t,T)$, $A_1(t,T)$ and $A_2(t,T)$ are given in Proposition \ref{Prop:solution1}. The first order condition \eqref{eq:fod} then becomes
\eqstar{
\alpha^{Y^\star}=&\frac{1}{\gamma} (\Sigma^\prime \Sigma)^{-1} My+\frac{1}{\gamma} \Sigma^{-1}  \xi Ay\\
=&\left[\begin{array}{ccc}
    \frac{\sigma_L^\lambda\sigma_S\theta_r\sqrt{r}-\sigma_L^r\sigma_S\theta_\lambda\sqrt{\lambda}-\sigma_L^\lambda\sigma_S^r\theta_S\sqrt{r}}{\sigma_B\sigma_S\sigma_L^\lambda}+\frac{\sigma_r\sqrt{r}}{\sigma_B}A_1-\frac{\sigma_L^r\sigma_\lambda\sqrt{\lambda}}{\sigma_B\sigma_L^\lambda}A_2\\
    \frac{\theta_\lambda\sqrt{\lambda}}{\sigma_L^\lambda}+\frac{\sigma_\lambda\sqrt{\lambda}}{\sigma_L^\lambda}A_2\\
    \frac{\theta_S}{\sigma_S} 
    \end{array}\right]\frac{y}{\gamma}.
}
\end{proof}
\end{appendix} 

\bibliographystyle{chicago}
\bibliography{longevityhedge}

\begin{thebibliography}{}

\bibitem[\protect\citeauthoryear{Battocchio and Menoncin}{Battocchio and
  Menoncin}{2004}]{battocchio2004optimal}
Battocchio, P. and F.~Menoncin (2004).
\newblock Optimal pension management in a stochastic framework.
\newblock {\em Insurance: Mathematics and Economics\/}~{\em 34\/}(1), 79--95.

\bibitem[\protect\citeauthoryear{Biffis and Blake}{Biffis and
  Blake}{2014}]{biffis2014keeping}
Biffis, E. and D.~Blake (2014).
\newblock Keeping some skin in the game: How to start a capital market in
  longevity risk transfers.
\newblock {\em North American Actuarial Journal\/}~{\em 18\/}(1), 14--21.

\bibitem[\protect\citeauthoryear{Biffis and Millossovich}{Biffis and
  Millossovich}{2006}]{biffis2006bidimensional}
Biffis, E. and P.~Millossovich (2006).
\newblock A bidimensional approach to mortality risk.
\newblock {\em Decisions in Economics and Finance\/}~{\em 29\/}(2), 71--94.

\bibitem[\protect\citeauthoryear{Blake and Burrows}{Blake and
  Burrows}{2001}]{blake2001survivor}
Blake, D. and W.~Burrows (2001).
\newblock Survivor bonds: Helping to hedge mortality risk.
\newblock {\em Journal of Risk and Insurance\/}, 339--348.

\bibitem[\protect\citeauthoryear{Boulier, Huang, and Taillard}{Boulier
  et~al.}{2001}]{boulier2001optimal}
Boulier, J.-F., S.~Huang, and G.~Taillard (2001).
\newblock Optimal management under stochastic interest rates: the case of a
  protected defined contribution pension fund.
\newblock {\em Insurance: Mathematics and Economics\/}~{\em 28\/}(2), 173--189.

\bibitem[\protect\citeauthoryear{Brigo and Mercurio}{Brigo and
  Mercurio}{2007}]{brigo2007interest}
Brigo, D. and F.~Mercurio (2007).
\newblock {\em Interest rate models-theory and practice: with smile, inflation
  and credit}.
\newblock Springer Science \& Business Media.

\bibitem[\protect\citeauthoryear{Cairns}{Cairns}{2000}]{cairns2000some}
Cairns, A. (2000).
\newblock Some notes on the dynamics and optimal control of stochastic pension
  fund models in continuous time.
\newblock {\em ASTIN Bulletin: The Journal of the IAA\/}~{\em 30\/}(1), 19--55.

\bibitem[\protect\citeauthoryear{Cairns, Blake, and Dowd}{Cairns
  et~al.}{2006}]{cairns2006two}
Cairns, A.~J., D.~Blake, and K.~Dowd (2006).
\newblock A two-factor model for stochastic mortality with parameter
  uncertainty: theory and calibration.
\newblock {\em Journal of Risk and Insurance\/}~{\em 73\/}(4), 687--718.

\bibitem[\protect\citeauthoryear{Cocco and Gomes}{Cocco and
  Gomes}{2012}]{cocco2012longevity}
Cocco, J.~F. and F.~J. Gomes (2012).
\newblock Longevity risk, retirement savings, and financial innovation.
\newblock {\em Journal of Financial Economics\/}~{\em 103\/}(3), 507--529.

\bibitem[\protect\citeauthoryear{Cuchiero}{Cuchiero}{2006}]{cuchiero2006affine}
Cuchiero, C. (2006).
\newblock {\em Affine interest rate models: theory and practice}.
\newblock na.

\bibitem[\protect\citeauthoryear{Dahl}{Dahl}{2004}]{dahl2004stochastic}
Dahl, M. (2004).
\newblock Stochastic mortality in life insurance: market reserves and
  mortality-linked insurance contracts.
\newblock {\em Insurance: mathematics and economics\/}~{\em 35\/}(1), 113--136.

\bibitem[\protect\citeauthoryear{De~Kort and Vellekoop}{De~Kort and
  Vellekoop}{2017}]{de2017existence}
De~Kort, J. and M.~Vellekoop (2017).
\newblock Existence of optimal consumption strategies in markets with longevity
  risk.
\newblock {\em Insurance: Mathematics and Economics\/}~{\em 72}, 107--121.

\bibitem[\protect\citeauthoryear{De~Moivre}{De~Moivre}{1725}]{de1725annuities}
De~Moivre, A. (1725).
\newblock Annuities on lives: Or, the valuation of annuities upon any number of
  lives as also of reversions.
\newblock {\em London: William Person\/}~{\em 1725}.

\bibitem[\protect\citeauthoryear{Deelstra, Grasselli, and Koehl}{Deelstra
  et~al.}{2003}]{deelstra2003optimal}
Deelstra, G., M.~Grasselli, and P.-F. Koehl (2003).
\newblock Optimal investment strategies in the presence of a minimum guarantee.
\newblock {\em Insurance: Mathematics and Economics\/}~{\em 33\/}(1), 189--207.

\bibitem[\protect\citeauthoryear{Duffee}{Duffee}{2002}]{duffee2002term}
Duffee, G.~R. (2002).
\newblock Term premia and interest rate forecasts in affine models.
\newblock {\em The Journal of Finance\/}~{\em 57\/}(1), 405--443.

\bibitem[\protect\citeauthoryear{DWP}{DWP}{2013}]{govuk}
DWP (2013).
\newblock Workplace pension reform: automatic enrolment earnings thresholds,
  review and revision 2012/2013.
\newblock
  \url{https://www.gov.uk/workplace-pensions/what-you-your-employer-and-the-government-pay}.
\newblock Accessed 5 May 2020.

\bibitem[\protect\citeauthoryear{Gao}{Gao}{2008}]{gao2008stochastic}
Gao, J. (2008).
\newblock Stochastic optimal control of dc pension funds.
\newblock {\em Insurance: Mathematics and Economics\/}~{\em 42\/}(3),
  1159--1164.

\bibitem[\protect\citeauthoryear{Gompertz}{Gompertz}{1825}]{gompertz1825xxiv}
Gompertz, B. (1825).
\newblock Xxiv. on the nature of the function expressive of the law of human
  mortality, and on a new mode of determining the value of life contingencies.
  in a letter to francis baily, esq. frs \&c.
\newblock {\em Philosophical transactions of the Royal Society of
  London\/}~{\em 115}, 513--583.

\bibitem[\protect\citeauthoryear{Guan and Liang}{Guan and
  Liang}{2014}]{guan2014optimal}
Guan, G. and Z.~Liang (2014).
\newblock Optimal management of dc pension plan in a stochastic interest rate
  and stochastic volatility framework.
\newblock {\em Insurance: Mathematics and Economics\/}~{\em 57}, 58--66.

\bibitem[\protect\citeauthoryear{Han and Hung}{Han and
  Hung}{2012}]{han2012optimal}
Han, N.-w. and M.-w. Hung (2012).
\newblock Optimal asset allocation for dc pension plans under inflation.
\newblock {\em Insurance: Mathematics and Economics\/}~{\em 51\/}(1), 172--181.

\bibitem[\protect\citeauthoryear{He and Liang}{He and
  Liang}{2013}]{he2013optimalb}
He, L. and Z.~Liang (2013).
\newblock Optimal investment strategy for the dc plan with the return of
  premiums clauses in a mean--variance framework.
\newblock {\em Insurance: Mathematics and Economics\/}~{\em 53\/}(3), 643--649.

\bibitem[\protect\citeauthoryear{He and Liang}{He and
  Liang}{2015}]{he2015optimal}
He, L. and Z.~Liang (2015).
\newblock Optimal assets allocation and benefit outgo policies of dc pension
  plan with compulsory conversion claims.
\newblock {\em Insurance: Mathematics and Economics\/}~{\em 61}, 227--234.

\bibitem[\protect\citeauthoryear{HMRC}{HMRC}{2018}]{hmrc}
HMRC (2018).
\newblock Guidance: Check if you have unused annual allowances on your pension
  savings.
\newblock
  \url{https://www.gov.uk/guidance/check-if-you-have-unused-annual-allowances-on-your-pension-savings}.
\newblock Accessed 5 May 2020.

\bibitem[\protect\citeauthoryear{Lee and Carter}{Lee and
  Carter}{1992}]{lee1992modeling}
Lee, R.~D. and L.~R. Carter (1992).
\newblock Modeling and forecasting us mortality.
\newblock {\em Journal of the American statistical association\/}~{\em
  87\/}(419), 659--671.

\bibitem[\protect\citeauthoryear{Luciano and Vigna}{Luciano and
  Vigna}{2005}]{luciano2005non}
Luciano, E. and E.~Vigna (2005).
\newblock Non mean reverting affine processes for stochastic mortality.
\newblock {\em ICER Applied Mathematics Working Paper\/}.

\bibitem[\protect\citeauthoryear{Menoncin}{Menoncin}{2008}]{menoncin2008role}
Menoncin, F. (2008).
\newblock The role of longevity bonds in optimal portfolios.
\newblock {\em Insurance: Mathematics and Economics\/}~{\em 42\/}(1), 343--358.

\bibitem[\protect\citeauthoryear{Menoncin}{Menoncin}{2009}]{menoncin2009death}
Menoncin, F. (2009).
\newblock Death bonds with stochastic force of mortality.
\newblock In {\em Actuarial and Financial Mathematics Conference--Interplay
  between finance and insurance}.

\bibitem[\protect\citeauthoryear{Menoncin and Regis}{Menoncin and
  Regis}{2017}]{menoncin2017longevity}
Menoncin, F. and L.~Regis (2017).
\newblock Longevity-linked assets and pre-retirement consumption/portfolio
  decisions.
\newblock {\em Insurance: Mathematics and Economics\/}~{\em 76}, 75--86.

\bibitem[\protect\citeauthoryear{Merton}{Merton}{1969}]{merton1969lifetime}
Merton, R.~C. (1969).
\newblock Lifetime portfolio selection under uncertainty: The continuous-time
  case.
\newblock {\em The review of Economics and Statistics\/}, 247--257.

\bibitem[\protect\citeauthoryear{Milevsky and Promislow}{Milevsky and
  Promislow}{2001}]{milevsky2001mortality}
Milevsky, M.~A. and S.~D. Promislow (2001).
\newblock Mortality derivatives and the option to annuitise.
\newblock {\em Insurance: Mathematics and Economics\/}~{\em 29\/}(3), 299--318.

\bibitem[\protect\citeauthoryear{OECD}{OECD}{2019}]{oecd}
OECD (2019).
\newblock Pensions at a glance 2019: Oecd and g20 indicators.
\newblock \url{https://doi.org/10.1787/b6d3dcfc-en}.
\newblock Accessed 5 May 2020.

\bibitem[\protect\citeauthoryear{Pham}{Pham}{2009}]{pham2009continuous}
Pham, H. (2009).
\newblock {\em Continuous-time stochastic control and optimization with
  financial applications}, Volume~61.
\newblock Springer Science \& Business Media.

\bibitem[\protect\citeauthoryear{Renshaw and Haberman}{Renshaw and
  Haberman}{2006}]{renshaw2006cohort}
Renshaw, A.~E. and S.~Haberman (2006).
\newblock A cohort-based extension to the lee--carter model for mortality
  reduction factors.
\newblock {\em Insurance: Mathematics and economics\/}~{\em 38\/}(3), 556--570.

\bibitem[\protect\citeauthoryear{Russo, Giacometti, Ortobelli, Rachev, and
  Fabozzi}{Russo et~al.}{2011}]{russo2011calibrating}
Russo, V., R.~Giacometti, S.~Ortobelli, S.~Rachev, and F.~J. Fabozzi (2011).
\newblock Calibrating affine stochastic mortality models using term assurance
  premiums.
\newblock {\em Insurance: mathematics and economics\/}~{\em 49\/}(1), 53--60.

\end{thebibliography}
\nocite{}

\end{document}